\documentclass[fontsize=12pt]{article}
\usepackage{amsmath,amssymb,amsthm}
\usepackage[longnamesfirst,numbers]{natbib}
\usepackage{mathabx}
\usepackage{setspace}
\usepackage{enumerate}         
\usepackage{bm}  
\usepackage{xfrac}

\usepackage{graphicx}
\usepackage[section]{placeins} 

\newtheorem{thm}{thm}[section]
\newtheorem{cor}[thm]{cor}
\newtheorem{lem}[thm]{lem}
\newtheorem{prop}[thm]{prop}

\theoremstyle{definition}
\newtheorem{defn}[thm]{Definition}

\newtheorem{ass}[thm]{Assumption}
\theoremstyle{remark}
\newtheorem{rem}[thm]{Remark}
\newtheorem{example}[thm]{Example}
\numberwithin{equation}{section}

\allowdisplaybreaks

\newcommand{\set}[1]{\left\{#1\right\}}
\newcommand{\Ind}[1]{\mathbf{1}_{\left\{#1\right\}}}
\newcommand{\RR}{\mathbb{R}}

\newcommand{\PP}{\mathbb{P}}

\newcommand{\NN}{\mathbb{N}}

\newcommand{\cF}{\mathcal{F}}

\newcommand{\Rplus}{\mathbb{R}_{\geqslant 0}}

\newcommand{\Rminus}{\mathbb{R}_{\leqslant 0}}
\newcommand{\pd}[2]{\frac{\partial #1}{\partial #2}}

\newcommand{\argmin}{\operatorname{argmin}}

\newcommand{\wt}[1]{{\widetilde{#1}}}
\newcommand{\Econd}[2]{\mathbb{E}\left[\left.#1\right|#2\right]}        
\newcommand{\E}[1]{\mathbb{E}\left[#1\right]}                     
\newcommand{\wh}[1]{{\widehat{#1}}}

\newcommand{\Tmax}{T_\text{max}}

\newcommand{\wmax}{w_\text{max}}

\title{Affine forward variance models}

\author{Jim Gatheral, Baruch College, CUNY,\\ {\tt jim.gatheral@baruch.cuny.edu},\\  \\ Martin Keller-Ressel, TU Dresden,\\
 {\tt Martin.Keller-Ressel@tu-dresden.de }
}

\date{\today}
\begin{document}

\maketitle

\begin{abstract}
We introduce the class of affine forward variance (AFV) models of which both the conventional Heston model and the rough Heston model are special cases. We show that AFV models can be characterized by the affine form of their cumulant generating function, which can be obtained as solution of a convolution Riccati equation. We further introduce the class of affine forward order flow intensity (AFI) models, which are structurally similar to AFV models, but driven by jump processes, and which include Hawkes-type models. We show that the cumulant generating function of an AFI model satisfies a generalized convolution Riccati equation and that a high-frequency limit of AFI models converges in distribution to the AFV model. 
\end{abstract}
\thanks{MKR gratefully acknowledges financial support from DFG grants ZUK~64 and KE~1736/1-1. We thank Masaaki Fukasawa and two anonymous referees for their insightful comments.}



\section{Introduction}

The class of affine processes, introduced in \cite{duffie2003affine}, consists of all continuous-time Markov processes taking values in $\Rplus ^m \times \RR^n$, whose log-characteristic function depends in an affine way on the initial state vector of the process.  Affine processes have proved particularly convenient for financial modeling, typically giving rise to models with tractable formulae for the values of financial claims; the perennially popular Heston model \cite{heston1993closed} is just one (and perhaps the most famous) example of such a model.

In this paper, we introduce the class of {\em affine forward variance (AFV) models} of which classical Markovian affine stochastic volatility models turn out to be a special case. By writing our model in forward variance form, we are able to provide a unique characterization of a much wider class of affine stochastic volatility models, which includes non-Markovian models, such as the rough Heston model of \cite{euch2018characteristic} or, more generally, stochastic volatility models driven by affine Volterra processes in the sense of \cite{jaber2017affine}. Our contribution is to provide necessary and sufficient conditions for a (non-Markovian) stochastic volatility model to be affine, thus adding a reverse direction to the results of \cite{jaber2017affine}, and with a simpler proof\footnote{Some of these simplifications are due to the fact that we limit ourselves to the (real-valued) moment generating function, as opposed to the (complex-valued) characteristic function, studied in \cite{euch2018characteristic}}. In essence, the rough Heston model is -- up to a choice of kernel -- the \emph{only} stochastic volatility model with an affine moment generating function.

Inspired by the original derivation \cite{euch2018characteristic} of the rough Heston model as a limit of simple pure jump models of order flow, we further introduce the class of {\em affine forward order flow intensity} (AFI) models. These model are structurally similar to affine forward variance models and generalize the simple order flow model of \cite{euch2018characteristic}, by allowing arbitrary order size distributions and more general decay of the self-excitation of order flow. We define a high-frequency limit in which such models give rise to continuous affine forward variance models. In so doing, we generalize and simplify previous such derivations. Moreover, there is a clear structural analogy between the results we prove for AFV and AFI models, adding insight to and generalizing the connection between microstructural models of order flow and stochastic volatility models first brought to light in \cite{jaisson2015limit} and \cite{jaisson2016rough}.

Our paper proceeds as follows.  In Section \ref{sec:AFVCM}, we introduce the class of affine forward variance models and show that a forward variance model has an affine cumulant generating function (CGF) if and only if it can be written in a very specific form.  We further show that the CGF can be obtained as the unique global solution of a convolution Riccati equation closely related to the Volterra-Riccati equations of \cite{jaber2017affine}.  In Section \ref{sec:AAFIM}, we introduce the class of AFI models, showing that the CGF of such models solves a \emph{generalized} convolution Riccati equation.  In Section \ref{sec:HF}, we show that AFI models become AFV models in a high-frequency limit, where order arrivals are extremely frequent and order sizes extremely small.

\section{Affine forward variance models}\label{sec:AFVCM}
\subsection{Forward variance models}
Let a probability space $(\Omega, \cF, \PP)$ with right-continuous filtration $(\cF_t)_{t \geq 0}$ and two independent, adapted Brownian motions $W$ and $W^\bot$ be given.  The filtration generated by $W$ only is denoted by $(\cF_t^W)_{t \geq 0}$. Our starting point is a generic stochastic volatility model $(S,V)$, where spot volatility $V$ is modeled by a $\cF^W$-adapted continuous, integrable, and non-negative process and the price process $S$ by
\begin{equation}\label{eq:SDE_S}
dS_t = S_t \sqrt{V_t} \left(\rho dW_t + \sqrt{1 - \rho^2} dW_t^\bot\right),
\end{equation}
for some correlation parameter $\rho \in [-1,1]$. Crucially, we \emph{do not} assume that $V$ is an Ito-process or even a semi-martingale. Instead, we focus on the family (indexed by $T > 0$) of forward variance processes
\begin{equation}\label{eq:fw_var}
\xi_t(T) := \Econd{V_T}{\cF_t} = \Econd{V_T}{\cF_t^W},
\end{equation}
which are, by definition, $\cF^W$-adapted martingales with terminal values $\xi_T(T) = V_T$. By the martingale representation theorem, there exists, for each $T > 0$, a predictable process $\eta_t(T)$ with $\int_0^T \eta_s(T)^2 ds < \infty$ a.s., such that
\begin{equation}\label{eq:SDE_xi}
d\xi_t(T) = \eta_t(T) dW_t, \quad t \in [0,T].
\end{equation}
We refer to \eqref{eq:SDE_S} together with \eqref{eq:SDE_xi} as stochastic volatility model in forward variance form, or simply as \emph{forward variance model}.
It is often convenient to use the log-price $X = \log S$ instead of $S$ and we note that it follows from \eqref{eq:SDE_S} that $X$ satisfies the SDE
\begin{equation}\label{eq:SDE_X}
dX_t = -\frac{V_t^2}{2} dt + \sqrt{V_t} \left(\rho dW_t + \sqrt{1 - \rho^2} dW_t^\bot\right).
\end{equation}
We also refer to $X$ together with the family of processes $(\xi_.(T))_{T > 0}$ as forward variance model and denote it by $(X,\xi)$.
Moreover, we use the following convention:
\[\text{For $t > T$, define } \begin{cases}\xi_t(T) := V_T\\ \eta_t(T) := 0. \end{cases} \]
which is consistent with \eqref{eq:fw_var} and allows to extend \eqref{eq:SDE_xi} to all $t \ge 0$.

Finally we introduce the following assumptions on the integrands $\eta_.(T)$:
\begin{ass}\label{ass:sigma}
\begin{enumerate}[(a)]
\item For $dt \otimes d\PP$-almost all $(t,\omega)$ it holds that $\tau \mapsto \eta_t(t+\tau,\omega)$ is non-negative, decreasing and continuous on $(0,\infty)$.
\item For any $T >0$ the integrability condition 
\begin{equation}\label{eq:int}
\int_0^T \left(\int_0^T \eta_r(s)^2 dr \right)^{1/2} ds < \infty
\end{equation}
holds almost surely.
\end{enumerate}
\end{ass}

We will show that in the class of \emph{affine} forward variance models, the integrand $\eta_t(T)$ must factor as $\eta_t(T) = \sqrt{V_t} \kappa(T-t)$, where $\kappa$ is a deterministic function. To describe the admissible functions we introduce the following:
\begin{defn}\label{def:kernel}
For $1 \le p < \infty$, an $\emph{$L_p$-kernel}$ is a function $\kappa: \Rplus \to \Rplus$ which is continuous on $(0,\infty)$ and satisfies $\int_0^T \kappa(t)^p dt < \infty$ for all $T > 0$.
\end{defn}

\begin{rem}
\begin{enumerate}[(a)]
\item If $\eta_t(T)$ factors as $\eta_t(T) = Z_t \kappa(T-t)$ into a non-negative continuous stochastic process $Z$ and a deterministic function $\kappa$, then Assumption~\ref{ass:sigma} is equivalent to $\kappa$ being a decreasing $L_2$-kernel in the sense of Definition~\ref{def:kernel}.
\item By \eqref{eq:SDE_S} $S$ is a non-negative local martingale and therefore a supermartingale, such that $\E{S_t} = \E{e^{X_t}} \le S_0$. This implies by Jensen's inequality that the moments $\E{S_t^u}$ are finite for all $u \in [0,1]$; a fact that will be used subsequently.
\end{enumerate}
\end{rem}

\subsection{A characterization of affine forward variance models}

\begin{defn}\label{def:affine}
We say that a forward variance model $(X,\xi)$ has an \emph{affine cumulant generating function} determined by $g(t,u)$, if its conditional cumulant generating function is of the form 
\begin{equation}\label{eq:affine}
\log \Econd{e^{u(X_T - X_t)}}{\cF_t} = \int_t^T g(T-s,u) \xi_t(s)ds,
\end{equation}
for all $u \in [0,1]$, $0 \le t \le T$ and where $g(.,u)$ is $\Rminus$-valued and continuous on $[0,T]$ for all $T > 0$ and $u \in [0,1]$.
\end{defn}
\begin{rem}
Alternatively, we could consider \eqref{eq:affine} with imaginary parameter $u = iz$ for $z \in \RR$, i.e. an \emph{affine log-characteristic function} as in \cite{euch2018characteristic}. This is of particular interest for applications of Fourier pricing to the model $(X,\xi)$. To show results on distributional convergence, which will be the subject of Section~\ref{sec:HF}, using the cumulant generating function is sufficient and it will turn out that restricting to real parameters simplifies many of the mathematical arguments.
\end{rem}

Convolution integrals, as in the exponent of \eqref{eq:affine}, will appear frequently in the following calculations and so it is natural to introduce the convolution operation $(f \star g)(t) := \int_0^t f(t-s)g(s)ds$. For functions with multiple arguments or subscripts, we use the convention that convolution acts on the first argument, excluding subscripts. Other arguments or subscripts are passed on to the result. With this convention \eqref{eq:affine} can be written succinctly as
\begin{equation}\label{eq:affine_compact}
\Econd{e^{u(X_T - X_t)}}{\cF_t} = \exp \Big( (g \star \xi)_t (T,u)\Big).
\end{equation}
The following result gives a characterization of all forward variance models with affine CGF. Its proof is given in Section~\ref{sec:proof} with some parts relegated to Appendix~\ref{app:proof}.

\begin{thm}\label{thm:main}
Under Assumption~\ref{ass:sigma} a forward variance model $(X,\xi)$ has an affine CGF if and only if 
\begin{equation}\label{eq:beta_sigma}
\eta_t(T) = \sqrt{V_t} \kappa(T-t), 
\end{equation}
for a deterministic, decreasing $L_2$-kernel $\kappa$. Moreover, $g(.,u): \Rplus \to \Rminus$ in \eqref{eq:affine} is the unique global continuous solution of the convolution Riccati equation
\begin{equation}\label{eq:Riccati}
g(t,u) = R_V\Big(u, \int_0^t \kappa(t-s) g(s,u) ds\Big) = R_V\Big(u,(\kappa \star g) (t,u)\Big),\quad t \ge 0
\end{equation}
where
\begin{equation}\label{eq:RV}
R_V(u,w) = \frac{1}{2}(u^2 - u)  +  \rho u w  + \frac{1}{2} w^2.
\end{equation}
\end{thm}
\begin{rem}\label{rem:outer_Riccati}
Alternatively, $g(t,u)$ can be written as
\[g(t,u) = R_V(u,f(t,u)),\]
where $f(t,u)$ is the unique global continuous solution of the non-linear Volterra equation
\begin{equation}\label{eq:Riccati_outer}
f(t,u) = \int_0^t \kappa(t-s) R_V(u,f(s,u)) ds.
\end{equation}
See Appendix~\ref{app:volterra} for further discussion of non-linear Volterra equations and for the equivalence of equations \eqref{eq:Riccati} and \eqref{eq:Riccati_outer}.
\end{rem}

We introduce the useful notion of a $\gamma$-resolvent kernel:
\begin{lem}\label{lem:resolvent}
Let $1 \le p < \infty$ and let $\kappa$ be an $L_p$-kernel. Then for any $\gamma \ge 0$, there exists a unique $L_p$-kernel $r$, such that
\begin{equation}\label{eq:resolvent}
r - \kappa = \gamma (\kappa \star r).
\end{equation}
If $\log \kappa$ is convex, then the assertion also holds for $\gamma < 0$. We call $r$ the \emph{$\gamma$-resolvent} of $\kappa$.
\end{lem}
\begin{proof}
If $\gamma = 0$ it clear that $r = \kappa$, and the Lemma becomes trivial. In all other cases $-\gamma r$ is the so-called `resolvent of the second kind' (see \cite{gripenberg1990volterra}) of $-\gamma \kappa$, and the properties of $r$ follow directly from Thm.~2.3.1 (existence and uniqueness), Thm.~2.3.5 (local $L_p$-integrability), Prop.~9.5.7 (continuity), Prop.~9.8.1 (positivity for $\gamma \ge 0)$ and Prop.~9.8.8 (positivity under log-convexity for $\gamma < 0$); all in \cite{gripenberg1990volterra}.
\end{proof}
Upon taking Laplace transforms $\wh{\kappa}, \wh{r}$ of $\kappa, r$, the resolvent equation \eqref{eq:resolvent} becomes
\begin{equation}\label{eq:resolvent_laplace}
\wh{r} - \wh{\kappa}  = \gamma (\wh{\kappa} \cdot  \wh{r}),
\end{equation}
which can be useful for determining $r$ explicitly from a given $\kappa$.

\begin{rem}[Relation to affine Volterra processes]\label{rem:volterra}
Note that the instantaneous variance process $V_t = \xi_t(t)$ of an AFV model can be written as
\begin{equation}\label{eq:affine_volterra}
V_t = \xi_0(t) + \int_0^t \kappa(t-s)\sqrt{V_s}dW_s,
\end{equation}
which shows that $V$ is an affine Volterra process in the sense of \cite{jaber2017affine}.  We emphasize, that the representation \eqref{eq:affine_volterra} is not unique. Indeed, let $\lambda > 0$ and let $\varphi$ be the $\lambda$-resolvent of $\kappa$. Then, using e.g. \cite[Lem.~2.5]{jaber2017affine}, it follows that
\begin{equation}\label{eq:affine_volterra_drift}
V_t = \tilde{\xi}_0(t)  - \lambda \int_0^t \varphi(t-s) V_s ds  + \int_0^t \varphi(t-s) \sqrt{V_s} dW_s,
\end{equation}
with $\tilde{\xi}_0 = \xi_0 + \lambda (\varphi \star \xi_0)$, and a mean-reverting drift-term appears. 
Conversely, if a process of the form \eqref{eq:affine_volterra_drift} is given, and if $\varphi$ has a $(-\lambda)$-resolvent $\kappa$, then the forward variance is of the form
\[\xi_t(T) = \Econd{V_T}{\cF_t} = \xi_0(T) + \int_0^t \kappa(T-s) \sqrt{V_s}dB_s,\]
with $\xi_0 = \tilde{\xi_0} - \lambda (\kappa \star \tilde{\xi}_0)$, cf. \cite[Lem.~2.5]{jaber2017affine}. 
\end{rem}

\subsection{Two examples: Heston and rough Heston models}

\begin{example}[The Heston model]\label{ex:heston}
The Heston model \cite{heston1993closed} is given by
\begin{subequations}\label{eq:heston}
\begin{align}
dS_t &= S_t \sqrt{V_t} \left(\rho dW_t + \sqrt{1 - \rho^2} dW_t^\bot\right)\label{eq:heston1}\\
dV_t &= -\lambda (V_t - \theta) dt + \zeta \sqrt{V_t} dW_t.\label{eq:heston2}
\end{align}
\end{subequations}
A simple calculation shows that
\[\xi_t(T) = \Econd{V_T}{\cF_t} = \theta \left(1 - e^{-\lambda(T-t)}\right) + e^{-\lambda (T-t)} V_t.\]
Hence,
\[d \xi_t(T) = \zeta e^{-\lambda(T-t)} \sqrt{V_t} dW_t\]
and it follows that the Heston model can be written as an affine forward variance model with kernel
\[\kappa(x) = \zeta e^{-\lambda x}\]
and initial forward variance
\[\xi_0(T) = \theta \left(1 - e^{-\lambda T}\right) + e^{-\lambda  T} V_0 = V_0 + (\theta - V_0) \lambda \int_0^T \kappa(s) ds.\]
Note that $\kappa$ is the $(-\lambda/\zeta)$-resolvent of the constant kernel $\zeta$, in accordance with Remark~\ref{rem:volterra}.
To obtain the Riccati ODEs for the Heston model in the usual form (see e.g. \cite{keller-ressel2011moment}), let $\psi(.,u)$ be a $C^1$-function such that 
\[g(t,u) = \pd{}{t}\psi(t,u) + \lambda \psi(t,u), \quad \text{and} \quad \psi(0,u) = 0.\] 
By partial integration we obtain
\[(\kappa \star g)(t,u) = \zeta \int_0^t e^{-\lambda (t-s)} g(s,u) = \zeta \psi(t,u).\]
Inserting into the convolution Riccati equation \eqref{eq:Riccati} yields
\[\pd{}{t}\psi(t,u) = \frac{1}{2}(u^2 - u) + (\zeta \rho u - \lambda) \psi(t,u) + \frac{\zeta^2}{2} \psi(t,u)^2, \qquad \psi(0,u) = 0,\]
in accordance with \cite{keller-ressel2011moment}. Furthermore, it is straightforward to show that 
\[\int_t^T g(t-s,u) \xi_t(s) ds = \phi(T-t,u) + V_t \psi(T-t,u),\]
with $\phi(t,u) = \lambda \theta \int_0^t \psi(s,u) ds$. \hfill $\diamond$
\end{example}

\begin{example}[The rough Heston model]\label{ex:rough}
In the rough Heston model, introduced in \cite{euch2018characteristic}, \eqref{eq:heston2} is replaced by 
\begin{equation}
V_t = V_0 + \frac{1}{\Gamma(\alpha)} \int_0^t (t-s)^{\alpha - 1} \lambda (\theta - V_s) ds + \frac{\zeta}{\Gamma(\alpha)}  \int_0^t (t-s)^{\alpha -1} \sqrt{V_s} dW_s
\end{equation}
where $\alpha \in (1/2,1)$ is related to the `roughness' of the paths of $V$. Note that this is an affine Volterra process \eqref{eq:affine_volterra_drift} with power-law kernel $\phi_\text{pow}(t) = \zeta t^{\alpha-1}/\Gamma(\alpha)$. In \cite{euch2018perfect} it is shown that the forward variance in the rough Heston model satisfies
\[d \xi_t(T) = \kappa(T-s) \sqrt{V_t} dW_t,\]
with the kernel
\[\kappa(x) = \zeta x^{\alpha - 1} E_{\alpha,\alpha}(-\lambda x^\alpha)\]
and where $E_{\alpha,\beta}(x)$ denotes the generalized Mittag-Leffler function (cf. \cite{erdelyi1981higher}, \cite[Sec.~1.2]{podlubny1998fractional}).\footnote{Alternatively, one can check that the Laplace transforms $\wh{\phi}_\text{pow}(z) = \zeta z^{-\alpha}$ and $\wh{\kappa}(z) = \zeta/(z^\alpha + \lambda)$ satisfy the relation \eqref{eq:resolvent_laplace} with $\gamma = -\lambda/\zeta$.} Thus, the rough Heston model is an affine forward variance model in the sense of Theorem~\ref{thm:main}. The initial forward variance is given by (cf. \cite[Prop.~3.1]{euch2018perfect})
\[\xi_0(T) = V_0 + (\theta - V_0) \lambda \int_0^T \kappa(s) ds.\]
To obtain the \emph{fractional} Riccati equation (cf. \cite[Eq.~(24)]{euch2018characteristic}) for the rough Heston model set
\[\psi(t,u) = \frac{1}{\zeta} (\kappa \star g)(t,u) = \frac{1}{\Gamma(\alpha)}\int_0^t (t-s)^{\alpha - 1} E_{\alpha,\alpha}(-\lambda(t-s)) g(s,u) ds.\]
By \cite[Lem.~A.2]{euch2018characteristic} $\psi(t,u)$ satisfies
\[D^\alpha \psi(t,u) + \lambda \psi(t,u) = g(s,u)\]
where $D^\alpha$ denotes the Riemann-Liouville fractional derivative of order $\alpha$.  
Inserting into the convolution Riccati equation \eqref{eq:Riccati} yields
\[D^\alpha \psi(t,u) = \frac{1}{2}(u^2 - u) + (\zeta \rho u - \lambda) \psi(t,u) + \frac{\zeta^2}{2} \psi(t,u)^2, \qquad \psi(t;0) = 0,\]
in accordance with \cite[Eq~(24)]{euch2018characteristic}. Denote by $I^\alpha f = \frac{1}{\Gamma(\alpha)} \int_0^\infty (t-s)^{\alpha - 1}f(s) ds$ the Riemann-Liouville fractional integral of order $\alpha$ and write $\bm{1}$ for the function of constant value one. The exponent in \eqref{eq:affine} can be transformed as follows:
\begin{align*}
\int_0^. g(.-s,u) \xi_0(s) ds &= V_0 (g \star \bm{1}) + (\theta - V_0) \frac{\lambda}{\zeta} (g \star \kappa \star \bm{1}) = \\
&= V_0 \left((g - \lambda\, \psi) \star \bm{1}\right) + \theta \lambda \, (\psi \star \bm{1}) = \\
&= V_0 (D^\alpha \psi) \star \bm{1} + \theta \lambda \int_0^. \psi(s,u) = \\
&= V_0 I^{1-\alpha} \psi + \theta \lambda \int_0^. \psi(s,u),
\end{align*}
which is the same as \cite[Eq~(23)]{euch2018characteristic}. \hfill $\diamond$
\end{example}

\subsection{Proving the characterization result}\label{sec:proof}
To prepare for the proof of Theorem~\ref{thm:main}, we introduce the following notation: Given a family $g(.,u)_{u \in [0,1]}$ of continuous functions from $\Rplus$ to $\RR$, we set
\begin{align}
G_t &= (g \star \xi)_t(T,u)  = \int_t^T g(T-s,u) \xi_t(s) ds,\label{eq:Gdef}\\
M_t &= \exp \left(u X_t + G_t \right) \label{eq:Mdef}.
\end{align}
If $(X,\xi)$ has an affine CGF determined by $g(t,u)$ then it follows from \eqref{eq:affine} that $M$ is a martingale. Conversely, if $M$ is a martingale, then \eqref{eq:affine} follows by taking conditional expectations. Hence, the affine property of $(X,\xi)$ can be characterized in terms of the martingale property of $M$. In order to apply It\^o's formula to $M$ we represent $G$ as an It\^o process. The calculation is analogous to the drift computation in the Heath-Jarrow-Morton-model (cf. \cite[Ch.~6]{filipovic2009term}) and uses the stochastic Fubini theorem to interchange stochastic integral and Lebesgue integral.
\begin{lem}\label{lem:G}
Let $G$ and $M$ be given as in \eqref{eq:Gdef}, \eqref{eq:Mdef} and let Assumption~\ref{ass:sigma} on the forward variance model $(X,\xi)$ hold. Then $G$ can be written in It\^o process form as
\[G_t = \int_0^T g(T-s,u) \xi_0(s) ds - \int_0^t g(T-s,u) V_s ds + \int_0^t h_s(T,u) dW_s,\]
where
\begin{equation}\label{eq:v}
h_t(T,u) = (g \star \eta)_t(T,u) = \int_t^T g(T-r,u) \eta_t(r) dr.
\end{equation}
Moreover, $M$ is an It\^o process, which can be decomposed as
\begin{align}\label{eq:M_decomp}
\frac{dM_t}{M_t} &= \text{loc.{}mg.{}}  + \\
 &+ \left\{ \frac{1}{2}(u^2 - u) V_t - g(T-t,u) V_t + u \rho \sqrt{V_t} h_t(T,u) + \frac{1}{2} h_t(T,u)^2 \right\} dt. \notag
\end{align}
\end{lem}
\begin{proof}Fix $T > 0$ and $u \in [0,1]$. Following \cite[p.~94]{filipovic2009term} closely, we compute
\begin{align*}
G_t &= \int_t^T g(T-s,u) \xi_t(s) ds = \\ 
& = \int_t^Tg(T-s,u) \xi_0(s) ds + \int_t^T \int_0^t g(T-s,u) \eta_r(s) dW_r\, ds \stackrel{stoch.{}Fub.{}}{=} \\
 &=   \int_t^Tg(T-s,u) \xi_0(s) ds + \int_0^t \int_t^T g(T-s,u) \eta_r(s)  ds\,dW_r = \\
 &= \int_0^Tg(T-s,u) \xi_0(s) ds + \int_0^t \int_r^T g(T-s,u) \eta_r(s)  ds\, dW_r - \\
 & \phantom{=} - \int_0^t g(T-s,u) \xi_0(s) ds - \int_0^t \int_r^t g(T-s,u) \eta_r(s)  ds\, dW_r  \stackrel{stoch.{}Fub.{}}{=} \\
&=  \int_0^Tg(T-s,u) \xi_0(s) ds + \int_0^t \int_r^T g(T-s,u) \eta_r(s)  ds \,dW_r - \\
 & \phantom{=} - \int_0^t g(T-s,u) \underbrace{\left(\xi_0(s) ds + \int_0^s \eta_r(s)  dW_r \right)}_{=V_s} ds.
\end{align*}
To justify the application of the stochastic Fubini theorem in the form of \cite[Thm.~2.2]{veraar2012stochastic}, we have to check whether 
\[I_t := \int_t^T \left(\int_0^t g(T-s,u)^2 \eta_r(s)^2 dr \right)^{1/2} ds\]
is finite for all $t \in [0,T]$. Since $g(.,u)$ is continuous, there is a finite constant $g_*(u) := \sup_{t \in [0,T]} |g(t,u)|$. Using Assumption~\ref{ass:sigma} we find that 
\[I_t \le g_*(u)  \int_0^T \left(\int_0^T \eta_r(s)^2 dr \right)^{1/2} ds < \infty\]
for all $t \in [0,T]$ and the application of the stochastic Fubini theorem was justified.

To show \eqref{eq:M_decomp}, we apply It\^o's formula to $M_t = \exp \left(u X_t + G_t \right)$ and obtain, using \eqref{eq:SDE_S} and \eqref{eq:SDE_xi} that
\begin{align}
\frac{dM_t}{M_t} &= u\, dX_t + dG_t + \frac{u^2}{2} d[X,X]_t + u\, d[X,G]_t + \frac{1}{2} d[G,G]_t = \label{eq:M}\\
 &= \text{loc.{}mg.{}}  + \notag\\
 &+ \left\{ \frac{1}{2}(u^2 - u) V_t - g(T-t,u) V_t + u \rho \sqrt{V_t} h_t(T,u) + \frac{1}{2} h_t(T,u)^2 \right\}  dt,\notag
 \end{align}
as claimed.
\end{proof}

\begin{lem}\label{lem:Riccati}
Let $\kappa$ be a decreasing $L_1$-kernel. Then for any $u \in [0,1]$, the convolution Riccati equation \eqref{eq:Riccati} has a unique global continuous $\Rminus$-valued solution $g(.,u)$. Moreover, $g(t,u) = R_V(u,f(t,u))$, where $f(.,u)$ is the unique global continuous solution of \eqref{eq:Riccati_outer}. 
\end{lem}
\begin{proof}
Fix $u \in (0,1)$ and set $H_u(w) = R_V(u,w)$, with $R_V$ given by \eqref{eq:RV}. It is easy to check that $H_u$ satisfies all conditions of Corollary~\ref{cor:Volterra}, i.e., that 
\begin{itemize}
\item $H_u(w)$ is a finite, strictly convex function on $(-\infty,0]$ and satisfies $H_u(0) < 0$;
\item $H_u(w)$ has a single root $H_u(w_*(u)) = 0$ in $(-\infty,0]$.
\end{itemize}
Thus, we conclude from Corollary~\ref{cor:Volterra} the existence of a unique global continuous solution $g(t,u)$ to \eqref{eq:Riccati} for all $u \in (0,1)$. Moreover, $g(t,u) \le 0$ for all $(t,u) \in \Rplus \times (0,1)$ by estimate \eqref{eq:g_estim}. Adding the boundary cases $u \in \{0,1\}$ is trivial: Observe that $g(t,u) \equiv 0$ is a constant global solution of \eqref{eq:Riccati} for $u \in \set{0,1}$, which must be unique by \cite[Thm.~13.1.2]{gripenberg1990volterra}. Also the representation as $g(t,u) = R_V(u,f(t,u))$ follows directly from Corollary~\ref{cor:Volterra}.
\end{proof}

We are now prepared to prove the first part of Theorem~\ref{thm:main}.
\begin{proof}[Theorem~\ref{thm:main}, `if' part.]
Let $\eta_t(T) = \sqrt{V_t} \kappa(T-t)$ and fix some $(T,u) \in (0,\infty) \times (0,1)$. By Lemma~\ref{lem:Riccati}, the convolution Riccati equation \eqref{eq:Riccati} associated to this kernel $\kappa$ has a unique global continuous $\Rminus$-valued solution $g(.,u)$. Using this solution $g(.,u)$ we define the processes $G$ and $M$ as in \eqref{eq:Gdef} and \eqref{eq:Mdef}. Applying Lemma~\ref{lem:G}, we see that equation \eqref{eq:v} for $h_t(T,u)$ can be simplified due to the factorization of $\eta_t(T)$ to 
\[h_t(T,u) = (g \star \eta)_t(T,u) = \sqrt{V_t} \int_0^{T-t} g(T-t-s,u) \kappa(s) ds = \sqrt{V_t} \cdot (g \star \kappa)(T-t).\] 
Inserting into \eqref{eq:M_decomp} we obtain
\begin{equation}
\frac{dM_t}{M_t} = \text{loc.{}mg.{}}  + \Big\{R_V\big(u, (g \star \kappa)(T-t)\big) - g(T-t,u) \Big\} V_t dt.
\end{equation}
Since $g(.,u)$ solves the convolution Riccati equation \eqref{eq:Riccati} the $dt$-term vanishes and we have shown that $M$ is a local martingale. It remains to show that $M$ is a true martingale. To this end, observe that 
\begin{equation}\label{eq:Lp_bound}
|M_t|^{1/u} = \exp \left(\frac{1}{u} (u X_t + G_t)\right)\le \exp(X_t) = S_t
\end{equation}
for all $t \in [0,T]$.
But $S$ is a supermartingale, such that
\[\E{|M_t|^{1/u}} \le \E{S_t} \le S_0.\]
Setting $p = 1/u > 1$ this is the $L_p$-criterion for uniform integrability of $M$. We conclude that $M$ is a true martingale, and hence that 
\begin{equation}\label{eq:cond_M}
\Econd{e^{u X_T}}{\cF_t} = \Econd{M_T}{\cF_t} = M_t = \exp \left(u X_t + \int_t^T g(T-s,u) \xi_t(s)ds\right), 
\end{equation}
showing \eqref{eq:affine}  for all $u \in (0,1)$. Adding the boundary cases $u \in \{0,1\}$ is trivial. Since $g(t,u) \equiv 0$ in these cases, \eqref{eq:cond_M} follows immediately. 
\end{proof}

For the reverse implication of Theorem~\ref{thm:main} we distinguish the cases $\rho \le 0$ and $\rho > 0$. The proof follows the same structure in both cases, but $\rho > 0$ needs some additional technical arguments, which are relegated to Appendix~\ref{app:proof}.
\begin{proof}[Theorem~\ref{thm:main}, `only if' part in the case $\rho \ge 0$.]
By assumption, the forward variance model $(X,\xi)$ has an affine CGF in the sense of Definition~\ref{def:affine}. Using the associated function $g(.,u)$, we define the processes $G$ and $M$ as in \eqref{eq:Gdef} and \eqref{eq:Mdef}. Observe that due to \eqref{eq:affine}, $M$ has to be a martingale. In addition, note that Assumption~\ref{ass:sigma}, together with $g(.,u) \le 0$ yields that
\[h_t(T,u) = (g \star \eta)_t(T,u) \le 0\]
for all $t, T \ge 0, u \in [0,1]$. Applying Lemma~\ref{lem:G} and using the fact that $M$ is a martingale, we see that the $dt$-term in \eqref{eq:M_decomp} has to vanish, i.e. that
\begin{equation}\label{eq:quadratic}
\frac{1}{2}(u^2 - u) V_t - g(T-t,u) V_t + u \rho \sqrt{V_t} h_t(T,u) + \frac{1}{2} h_t(T,u)^2 = 0,
\end{equation}
up to a $(dt \times d\PP)$-nullset. This is a quadratic equation in the variable $h_t(T,u)$, and due to the assumption $\rho \le 0$ it possesses a unique negative solution
\[h_t(T,u) = \sqrt{V_t} q_-(T-t,u) := \sqrt{V_t} \big(-\rho u - \sqrt{u^2(\rho^2 - 1) + u + 2g(T-t,u)}\big).\]
Inserting the definition of $h_t(T,u)$ from \eqref{eq:v} and setting $\tau = T-t$ we obtain
\begin{equation}\label{eq:Volterra_first}
\int_0^\tau g(\tau-s,u) \eta_t(t+s) ds = \sqrt{V_t }q_-(\tau,u), \qquad \forall \tau \ge 0.
\end{equation}
This is a Volterra integral equation of the first kind (cf. \cite[Sec.~5.5]{gripenberg1990volterra}) for the unknown function $\eta_t(t+.)$. From \eqref{eq:quadratic} it can easily be seen that $g(0,u) = \frac{1}{2}(u^2 - u) < 0$. Therefore, by \cite[Ex.~5.25]{gripenberg1990volterra}, see also \cite[Thm.~1]{gripenberg1980volterra}, there exists a (locally finite Borel) measure $\pi(d\tau,u)$ --  the \emph{resolvent of the first kind}\footnote{Not to be confused with the $\gamma$-resolvent of Lemma~\ref{lem:resolvent}.} of $g(.,u)$ -- such that $\int_0^\tau g(\tau-s,u) \pi(ds,u) = 1$ for all $\tau > 0$. Convolving \eqref{eq:Volterra_first} with $\pi(d\tau,.)$, the unique solution of \eqref{eq:Volterra_first} can be expressed as, cf. \cite[Thm.~5.3]{gripenberg1990volterra},
\begin{equation}
\eta_t(t + \tau) = \sqrt{V_t} \cdot \pd{}{\tau} \left( \int_0^\tau q_-(\tau-s,u) \pi(ds,u) \right).
\end{equation}
Denoting the last factor by $\kappa(\tau)$ and taking into account Assumption~\ref{ass:sigma} the desired decomposition \eqref{eq:beta_sigma} follows. 
\end{proof}

\section{Affine forward order flow intensity models}\label{sec:AAFIM}
We now introduce a class of models for market order flow, which are structurally similar to forward variance models. These models consist of a log-price $X$ and a \emph{forward intensity process} $\xi_t(T)$, which models the expectation (at time $t$) of the future intensity of order flow (at time $T$). The forward intensity $\xi_t(T)$ has a role similar to the forward variance, and we call the resulting model an {\em affine forward order flow intensity} (AFI) model.\footnote{The strong empirical correlation between order volume (as a proxy for intensity) and return variance is well-documented in the literature (see e.g. \cite{gallant1992stock}). Therefore the parallels between AFV and AFI models should not come as a complete surprise.} The AFI model is driven purely by the arrival of market orders, which are represented by two pure-jump semimartingales $J^+_t, J^-_t$ of finite activity and with common intensity $\lambda_{t-}$. As in \eqref{eq:fw_var}, we assume that $\xi_.(T)$ and $\lambda$ are connected by 
\[\xi_t(T) := \Econd{\lambda_{T-}}{\cF_t}.\]
The driving processes $J^\pm$ jump only upwards and represent the arrival of buy and sell orders respectively. Their jump height distributions are given by two probability measures $\zeta_\pm(dx)$ on $\Rplus$ for buy and for sell orders. We assume that $\int_0^\infty e^x \zeta_\pm(dx) < \infty$; in particular, also the first moments
\[m_\pm := \int_0^\infty x \, \zeta_\pm(dx)\]
exist. In addition, we assume that the order flow processes are self-exciting, in the sense that each arriving order positively impacts the intensity process.  This impact can be asymmetric, i.e. the degree of self-excitement may be different for buy- and sell-orders. Together this leads to the specification
of the AFI model as
\begin{subequations}\label{eq:AFI}
\begin{align}  
dX_t &= -\lambda_{t-} m_X dt + dJ_t^+ - dJ_t^-,\label{eq:X_AFI}\\
d\xi_t(T) &=  \kappa(T-t)  \left(\gamma_+ d\wt{J}^+_t  + \gamma_- d\wt{J}^-_t\right).
\end{align}
\end{subequations}
where $\kappa$ is an $L_1$-kernel in the sense of Definition~\ref{def:kernel}, $\gamma_\pm$ are positive constants, $m_X$ is determined by the martingale condition on $S=e^X$ and $\wt{J}^{\pm}_t$ denote the \emph{compensated} order flow processes, i.e. $\wt{J}^{\pm}_t := J^\pm_t - m_\pm \int_0^t \lambda_{s-}ds$. Setting 
\begin{equation}\label{eq:j_compensated}
J_t^X = J_t^+ - J_t^-, \qquad 
J_t^\lambda = \gamma_+ J_t^+ + \gamma_- J_t^-
\end{equation}
and denoting by $\wt{J}^\lambda$ the compensated counterpart of $J^\lambda$, we can rewrite \eqref{eq:AFI} as
\begin{align*}  
dX_t &= -\lambda_{t-} m_X dt + dJ_t^X,\\
d\xi_t(T) &=  \kappa(T-t) d\wt{J}^\lambda_t.
\end{align*}
We proceed to discuss the jump processes and the compensators of their random jump measures in more detail. The random jumps of $J^\pm$ are compensated by
\[d\nu^\pm_t(dx) = \lambda_{t-} \zeta_\pm(dx) dt,\]
where $x$ represents jump size. While $J^+$ and $J^-$ are independent, given $\lambda$, it is important to note that $J_t^X$ and $J_t^\lambda$ are not. Instead, they move by simultaneous jumps. Thus, the predictable compensator of the jump measure of $(J^X, J^\lambda)$ is given by
\begin{align*}
d\nu_t^{(X,\lambda)}(dx,dy) &= \lambda_{t-} \chi(dx,dy) dt, \qquad \text{where}\\
\chi(dx,dy) &= \Big(\Ind{x \ge 0}\Ind{y = \gamma_+ x}\zeta_+(dx) +\Ind{x \le 0}\Ind{y = -\gamma_- x}\zeta_-(-dx)\Big).
\end{align*}
Note that the measure of joint jump heights $\chi(dx,dy)$ is concentrated on the line segments $y = \gamma_+ x,\;(x \ge 0)$ and $y = -\gamma_- x,\;(x \le 0)$ due to the simultaneity of jumps. In addition, we define 
\begin{equation}\label{eq:psi_def}
\psi_\pm(u) = \int_0^\infty \left(e^{ux}  - 1\right) \zeta_\pm(dx)
\end{equation}
and calculate
\[\int_{\RR \times \Rplus} \left(e^{ux + wy}  - 1 \right) \chi(dx,dy) = \psi_+\big(u + w \gamma_+\big) + \psi_-\big(-u + w \gamma_-\big).\]
Applying It\^o's formula for jump processes to $e^X$ it is easy to see that the martingale condition implies that
\[m_X = \psi_+(1) + \psi_-(-1).\]
The following theorem is the analogue of Theorem~\ref{thm:main} and shows the structural similarity between affine forward variance models and AFI models.

 \begin{thm}\label{thm:AFI}
The AFI model \eqref{eq:AFI} has an affine CGF in the sense of Definition~\ref{def:affine}. Moreover, $g(.,u): \Rplus \to \Rminus$ in \eqref{eq:affine} is the unique global solution of the generalized convolution Riccati equation
\begin{equation}\label{eq:Riccati_jump}
g(t,u) = R_\lambda\Big(u, \int_0^t \kappa(t-s) g(s,u) ds\Big) = R_\lambda\Big(u,(\kappa \star g) (t,u)\Big),
\end{equation}
where
\begin{equation}
R_\lambda(u,w) = \psi_+\big(u + w \gamma_+\big) + \psi_-\big(-u + w \gamma_-\big) -u m_X - w\big(\gamma_+ m_+ + \gamma_- m_-\big),
\end{equation}
with $\psi_\pm$ as in \eqref{eq:psi_def}.
\end{thm}

\begin{proof}Essentially, we proceed as in the `if'-part of the proof of Theorem~\ref{thm:main}. Let $G$ be defined as in \eqref{eq:Gdef} and set $M_t = \exp(u X_t + G_t)$. Applying the same argument as in the proof of Lemma~\ref{lem:G}, but replacing Brownian motion by the pure-jump-martingale $\wt{J}^\lambda$ we obtain
\[G_t = \int_0^T g(T-s,u) \xi_0(s) ds - \int_0^t g(T-s,u) \lambda_{s-} ds + \int_0^t h_s(T,u) d\wt{J}^\lambda_s,\]
where
\[h_t(T,u) = \lambda_{t-} \int_t^T \kappa(r-t) g(T- r,u) dr.\]
Applying the It\^o-formula with jumps to $M$ we obtain
\[M_t = M_0 + \int_0^t M_{s-} \left(u dX_t + dG_t \right)  + \sum_{0 \le s \le t}  M_{s-} \left(e^{ u \Delta X_s + \Delta G_s} - 1 - u \Delta X_s - \Delta G_s \right)\]
and compensating the jumps yields
\begin{align}\label{eq:martingale_jump}
\frac{dM_t}{M_t} &= \text{loc. mg.} - \lambda_{t-} u m_X dt + \big(\gamma_+ m_+ + \gamma_- m_-\big) h_t(T,u) dt  - g(T-t,u) \lambda_{t-} dt + 
\notag\\ &+ \lambda_{t-} \int_{\RR \times \Rplus} \left(e^{u x + y h_t(T,u)} - 1\right) \chi(dx,dy) dt = \notag\\
&= \text{loc. mg.} + \lambda_{t-} \Big\{ R_\lambda\big(u,(g \star \kappa)(T-t,u)\big) - g(T-t,u)\Big\} dt. 
\end{align}
where `loc. mg.' denotes a local martingale part that we need not compute explicitly. We see that the $dt$-terms vanish, if
\[g(\tau,u) = R_\lambda\Big(u, \int_0^\tau \kappa(\tau-s) g(s,u) ds \Big),\]
i.e. if the generalized convolution Riccati equation \eqref{eq:Riccati_jump} has a solution for $0 \le \tau \le T-t$. \\
To show that there exists a unique global continuous solution of \eqref{eq:Riccati_jump}, we proceed as in the proof of Lemma~\ref{lem:Riccati}, i.e., we set $H_u(w) = R_\lambda(u,w), u \in (0,1)$ and show that $H_u$ satisfies the conditions of Corollary~\ref{cor:Volterra}. In particular, for all $u \in (0,1)$, 
\begin{itemize}
\item $H_u(w)$ is a finite, strictly convex function on $(-\infty,0]$ and satisfies $H_u(0) < 0$;
\item $H_u(w)$ has a single root $H_u(w_*(u)) = 0$ in $(-\infty,0]$.
\end{itemize}
Indeed, note that strict convexity is inherited from $\psi_\pm$, cf. \eqref{eq:psi_def}. In addition, convexity of the exponential function implies, for $u \in (0,1)$, that
\[e^{ux} = e^{u \cdot x + (1-u) \cdot 0} \le u e^x + (1 - u) e^0 < u e^x + 1,\]
and hence that
\[H_u(0) = \int_0^\infty \left(e^{ux} - 1- ue^x\right)  \zeta_+(dx) + \int_0^\infty \left(e^{-ux} - 1- ue^{-x} \right) \zeta_-(dx) < 0.\]
Finally, the existence of the root $w_*(u)$ follows from the fact that
\[\lim_{w \to -\infty} \int_0^\infty \left( e^{(\pm u+\gamma_\pm w)x} - 1 - \gamma_\pm w x \right) \zeta_\pm(dx)  = + \infty,\]
which implies that $\lim_{w \to -\infty} H_u(w) = +\infty$. In summary, $H_u$ satisfies all conditions of Cor.~\ref{cor:Volterra} and we conclude the existence of a unique global solution $g(t,u)$ of the Riccati equation for all $u \in (0,1)$. Moreover, $g(t,u) \le 0$ for all $(t,u) \in \Rplus \times (0,1)$ by estimate \eqref{eq:g_estim}. We can add the boundary cases $u \in \{0,1\}$, observing that they yield the constant global solution $g(t,u) \equiv 0$, which must be unique by \cite[Thm.~13.1.2]{gripenberg1990volterra}. From \eqref{eq:martingale_jump} we conclude that $M_t = \exp(u X_t + G_t)$ is a local martingale.  
By the same arguments as in \eqref{eq:Lp_bound} and \eqref{eq:cond_M} it follows that $M$ is a true martingale, and hence that $(X,\xi)$ has an affine CGF.
\end{proof}

\begin{example}[The bivariate Hawkes process of \cite{euch2018characteristic}]\label{ex:hawkes}
Consider \eqref{eq:X_AFI}, driven by a bivariate Hawkes process $(J^+, J^-)$ with unit jump size (i.e., $\zeta_\pm(dx) = \delta_1(dx))$, common kernel $\varphi$, and common intensity $\lambda_t$, given by
\[\lambda_t = \mu + \int_0^t \varphi(t-s) \left(\gamma_+ dJ_s^+ + \gamma_- dJ_s^-\right),\]
as in \cite[Sec.~2]{euch2018characteristic}. Set $\widehat{\gamma}:= \gamma_+ + \gamma_-$ and let $\kappa$ be the $\widehat{\gamma}$-resolvent of $\varphi$ in the sense of \eqref{eq:resolvent}. In terms of $\kappa$, the Hawkes intensity $\lambda$ has the martingale representation (cf. \cite[Eq.~(45)]{bacry2015hawkes}) 
\[\lambda_t = \mu + \mu \widehat{\gamma} \int_0^t \kappa(t-u)  du + \int_0^t \kappa(t-u) d\tilde{J}^\lambda_u,\]
with the last integral now driven by a \emph{compensated} jump processes. Taking conditional expectations and using the martingale property of $\tilde{J}^\lambda$ yields
\[\Econd{\lambda_T}{\cF_t} = \mu +  \mu \widehat{\gamma} \int_0^T \kappa(T-u)du + \int_0^t \kappa(T-u) d\tilde{J}^\lambda_u,
\]
and hence
\[d\xi_t(T) = d\Econd{\lambda_T}{\cF_t} = \kappa(T-t)d\tilde{J}^\lambda_t,\]
which shows that the model can be cast as AFI model with kernel $\kappa$. For concrete specifications of $\varphi$, we can take Laplace transforms and use the relation \eqref{eq:resolvent_laplace} to determine the corresponding $\kappa$. Consider, for example
\begin{equation}\label{eq:laplace_exp}
\varphi(x) = e^{-(\lambda + \widehat{\gamma})x} \qquad \text{with Laplace tf.} \qquad \wh{\varphi}(z) = \frac{1}{z + \lambda + \widehat{\gamma}}.
\end{equation}
For this $\varphi$ we obtain from \eqref{eq:resolvent_laplace} the $\widehat{\gamma}$-resolvent
\[\kappa(x) = e^{-\lambda x} \qquad \text{with Laplace tf.} \qquad \wh{\kappa}(z) = \frac{1}{z + \lambda}, \]
i.e., the kernel of the Heston model in forward variance form; see Example~\ref{ex:heston}.
Furthermore, the Hawkes kernel
\begin{equation}\label{eq:kernel_ML}
\varphi(x) = x^{\alpha-1} E_{\alpha,\alpha}(-(\lambda  + \widehat{\gamma}) x^\alpha)
\end{equation}
has Laplace transform $\wh{\varphi}(z) = 1 / ( z^\alpha + \lambda + \widehat{\gamma})$ (cf. \cite[Eq.~(7.5)]{haubold2011mittag}). Thus its $\widehat{\gamma}$-resolvent is
\begin{equation}\label{eq:resolvent_ML}
\kappa(x) =  x^{\alpha-1} E_{\alpha,\alpha}(-\lambda x^\alpha),
\end{equation}
the kernel of the \emph{rough} Heston model in forward variance form; see Example~\ref{ex:rough}.
\end{example}

\section{High-frequency limit of the AFI model}\label{sec:HF}

We proceed to show that the AFV model is the high-frequency limit of the AFI model. This limit is closely related to the limits of `nearly unstable' Hawkes processes considered in \cite{jaisson2015limit, jaisson2016rough, euch2018characteristic}, see Example~\ref{ex:unstable} below.
\subsection{A first convergence result}
Our starting point is the AFI model \eqref{eq:X_AFI}. We assume that buy/sell order size distributions $\zeta_\pm$ are normalized in the sense that
\begin{equation}\label{eq:normalize}
\int_0^\infty x^2 \zeta_+(dx) + \int_0^\infty x^2 \zeta_-(dx) = 1,
\end{equation}
and we denote by 
\begin{equation}\label{eq:pdef}
p := \int_0^\infty x^2 \zeta_+(dx) \in [0,1]
\end{equation}
the variance of buy orders relative to sell orders. We introduce a small parameter $\epsilon$ and rescale \eqref{eq:AFI} as
\begin{subequations}\label{eq:AFI_scale}
\begin{align}  
dX^\epsilon_t &= -\lambda^\epsilon_{t-} m_X dt + dJ_t^{\epsilon,+} - dJ_t^{\epsilon,-},\\
d\xi^\epsilon_t(T) &=  \kappa^\epsilon(T-t)  \left(\gamma_+ d\wt{J}^{\epsilon,+}_t  + \gamma_- d\wt{J}^{\epsilon,-}_t\right),
\end{align}
\end{subequations}
where $J^{\epsilon,\pm}$ are pure jump semimartingales, independent given their common intensity $\lambda^\epsilon_t = \xi^\epsilon_t(t)$ 
and with jump height distribution
\[\zeta^\epsilon_\pm(dx) = \zeta_\pm(dx / \sqrt{\epsilon}).\]
Moreover, the kernels are scaled as 
\[\kappa^\epsilon(x) = \frac{1}{\epsilon} \kappa(x).\]
Thus, as $\epsilon \downarrow 0$, the frequency of jumps increases proportional to $1/\epsilon$, while the size of jumps shrinks proportional to $\sqrt{\epsilon}$. 
The initial conditions of \eqref{eq:AFI_scale} are given by $X^\epsilon_0 = X_0$ and $\xi^\epsilon_0(T) = \frac{1}{\epsilon} \xi_0(T)$. Under the given scaling, the quantities from \eqref{eq:psi_def} and below transform as
\begin{align*}
\psi_\pm^\epsilon(u) &= \psi_\pm(\sqrt{\epsilon} u)\\
m_X^\epsilon &= \psi_+(\sqrt{\epsilon}) + \psi_-(-\sqrt{\epsilon})\\
m^\epsilon_\pm &= \sqrt{\epsilon} m_\pm
\end{align*}
and we write
\[R^\epsilon(u,w) = \psi_+^\epsilon\big(u + w \gamma_+\big) + \psi_-^\epsilon\big(-u + w \gamma_-\big) -u m^\epsilon_X - w \big(\gamma_+ m_+^\epsilon + \gamma_- m_-^\epsilon\big).\]

\begin{lem}\label{lem:scale}
Given $\gamma_\pm > 0$ and the jump height distributions $\zeta_\pm(dx)$, define $c > 0$ and $\rho \in [-1,1]$ by
\begin{equation}\label{eq:rhodef}
c = \sqrt{p \gamma_+^2 + (1-p) \gamma_-^2} \qquad \rho = \frac{1}{c} \Big(p\gamma_+ - (1-p) \gamma_-\Big).
\end{equation}
Then
\[\lim_{\epsilon \to 0} \frac{1}{\epsilon} R^\epsilon(u,w) = \frac{1}{2}(u^2 - u)  +  c \rho u w  + \frac{c^2}{2} w^2 = R_V(u,cw)\]
with $R_V(u,w)$ as in \eqref{eq:RV}. Moreover, also the partial derivatives with respect to $u$ and $w$ converge, i.e. 
\begin{equation*}
\begin{split}
\lim_{\epsilon \to 0} \frac{1}{\epsilon} \pd{R^\epsilon}{u}(u,w) &= \pd{R_V}{u}(u,cw) = u - \frac{1}{2}  +  c \rho w\\
\lim_{\epsilon \to 0} \frac{1}{\epsilon} \pd{R^\epsilon}{w}(u,w) &= \pd{R_V}{w}(u,cw) = c w  +  \rho u.
\end{split}
\end{equation*}
\end{lem}
\begin{proof}
We can write
\begin{equation}\label{eq:R_epsilon_rep}
R^\epsilon(u,w) = \int_0^\infty b_+^\epsilon(x; u,w) \zeta_+(dx) + \int_0^\infty b_-^\epsilon(x; u,w) \zeta_-(dx)
\end{equation}
where
\begin{align*}
b_\pm^\epsilon(x,u,w) &= - u (e^{\pm \sqrt{\epsilon}x} -1) - w \gamma_\pm \sqrt{\epsilon} x  + \\ 
&+ \exp\Big((\pm u + w\gamma_\pm)\sqrt{\epsilon}x\Big) - 1.
\end{align*}
Expanding in powers of $\sqrt{\epsilon}x$ yields
\[b_\pm^\epsilon(x,u,w) = \epsilon x^2 \left(\frac{1}{2}(u^2 - u) \pm uw \gamma_\pm + \frac{w^2}{2} (\gamma_\pm)^2\right) + \mathcal{O}(\epsilon^{3/2}x^3).\]
Hence, using \eqref{eq:normalize} and \eqref{eq:pdef}, it follows that 
\begin{align*}
\lim_{\epsilon \to 0} \frac{1}{\epsilon}  R^\epsilon(u,w) &= \frac{1}{2} (u^2 - u) + uw (p \gamma_+  - (1-p)  \gamma_-) + \frac{w^2}{2} \left(p\gamma_+^2 + (1-p) \gamma_-^2\right) = \\ &= \frac{1}{2}(u^2 - u)  +  c \rho u w  + \frac{c^2}{2} w^2 = R_V(u,cw),
\end{align*}
where exchanging limit and integral is justified by dominated convergence and the integrability condition $\int_0^\infty e^x \zeta_\pm(dx) < \infty$.

To show the convergence of partial derivatives, we take partial derivatives in \eqref{eq:R_epsilon_rep} to obtain
\[\pd{R^\epsilon}{u}(u,w) = \int_0^\infty \pd{b_+^\epsilon}{u}(x; u,w) \zeta_+(dx) + \int_0^\infty \pd{b_-^\epsilon}{u}(x; u,w) \zeta_-(dx).\]
Since $R^\epsilon$ is convex, its difference quotients converge monotonically, and monotone convergence can be used to exchange derivative and integral. Expanding $\pd{b_\pm^\epsilon}{u}(x; u,w)$ in powers of $\sqrt{\epsilon}x$, a direct calculation yields the desired limit. The proof for the $\pd{}{w}$-derivative is analogous. 
\end{proof}

\begin{rem}
Equation~\eqref{eq:rhodef} gives important insights on the dependence of the leverage parameter $\rho$ on the micro-structural parameters $p$ (asymmetry of order sizes) and $\gamma_\pm$ (asymmetry of self-excitement). Consider first the case of symmetric order size distributions $p = \tfrac{1}{2}$ as in \cite{el2018microstructural}. In this case
\[\rho = \frac{\gamma_+ - \gamma_-}{\sqrt{2(\gamma_+^2 + \gamma_-^2)}},\]
which only takes values within $(-\sfrac{1}{\sqrt{2}}, \sfrac{1}{\sqrt{2}}) \approx (-0.71,0.71)$, with boundaries attained in the limiting cases $\gamma_\pm \to \infty$. When also asymmetry of order sizes is allowed, then $\rho$ can be represented as a scalar product of unit length vectors
\[\rho = \begin{pmatrix} \sqrt{p}, & \sqrt{1-p} \end{pmatrix} \cdot \begin{pmatrix}  \sfrac{\gamma_+\sqrt{p}}{c} \\  \sfrac{-\gamma_-\sqrt{1-p}}{c} \end{pmatrix},\]
which confirms that $\rho \in (-1,1)$. In addition, the boundary cases $\rho = \pm 1$ can be attained when the vectors are of opposing resp. of the same direction, i.e., when $p \to 0$ and $p \to 1$. 
\end{rem}

Combining Lemma~\ref{lem:scale} with Theorem~\ref{thm:main} and \ref{thm:AFI} yields a first distributional convergence result.

\begin{prop}\label{prop:convergence}
Let $(X^\epsilon, \xi^\epsilon)$ be the rescaled AFI model \eqref{eq:AFI_scale}. Define $c, \rho$ as in Lemma~\ref{lem:scale} and set $\kappa_V(x) = c \kappa(x)$. Then, for any $t \ge 0$, 
\begin{equation}\label{eq:dist_conv}
X_t^\epsilon \xrightarrow{\epsilon \to 0} X_t \qquad \text{in distribution},
\end{equation}
where $(X, \xi)$ is a forward variance model with correlation parameter $\rho$, and kernel $\kappa_V$. 
\end{prop}
\begin{proof}
By Theorem~\ref{thm:AFI}, $g^\epsilon(t,u)$ in the CGF \eqref{eq:affine} of $X^\epsilon$ is the unique global solution of the generalized convolution Riccati equation \eqref{eq:Riccati_jump} and hence satisfies
\begin{equation}\label{eq:gen_Riccati_scaled}
\tfrac{1}{\epsilon}g^\epsilon(t,u) = \tfrac{1}{\epsilon} R^\epsilon\Big(u, \kappa^\epsilon \star g^\epsilon(t,u)\Big) = \tfrac{1}{\epsilon} R^\epsilon\Big(u, \kappa  \star \left(\tfrac{1}{\epsilon}g^\epsilon\right)(t,u)\Big).
\end{equation}
Note that $\frac{1}{\epsilon}R^\epsilon(u,w)$ is jointly continuous in all variables, and by Lemma~\ref{lem:scale} converges to $R_V(u,cw)$ as $\epsilon \to 0$. By Corollary~\ref{cor:Volterra}, equation \eqref{eq:gen_Riccati_scaled} can be transformed into a non-linear Volterra equation of type \eqref{eq:Volterra}, whose solution depends jointly continuous on $(t,\epsilon,u)$ by \cite[Thm.~13.1.1]{gripenberg1990volterra}. We conclude that $\tfrac{1}{\epsilon}g^\epsilon(t,u)$ converges, uniformly for $(t,u)$ in compacts, to $g(t,u)$ as $\epsilon \to 0$, where $g(t,u)$ is the unique solution (cf. Theorem~\ref{thm:main}) of 
\[g(t,u) = R_V\Big(u, c \kappa \star g(t,u)\Big) = R_V\Big(u, \kappa_V \star g(t,u)\Big).\]
Using Theorems~\ref{thm:main} and \ref{thm:AFI}, we conclude that
\begin{multline*}
\E{e^{uX^\epsilon_t}} = \exp\left(u X_0 + \int_0^t g^\epsilon(t-s,u)\xi^\epsilon_0(s)ds\right) \to\\ \exp\left(u X_0 + \int_0^t g(t-s,u)\xi_0(s)ds\right) = \E{e^{uX_t}},
\end{multline*}
as $\epsilon  \to 0$, i.e., the moment generating function of $X^\epsilon_t$ converges to the moment generating function of $X_t$ on $u \in [0,1]$. By \cite[Prob.~30.4]{billingsley1986probability}, convergence of moment generating functions on a (non-empty) interval implies the convergence in distribution in \eqref{eq:dist_conv}.
\end{proof}

The following example shows that the scaling in \eqref{eq:AFI_scale} is related to the `nearly unstable' limit of Hawkes models in \cite{euch2018characteristic}.
\begin{example}[Nearly unstable limit of bivariate Hawkes processes]\label{ex:unstable}We continue Example~\ref{ex:hawkes} and consider the bivariate Hawkes process from \cite{euch2018characteristic} with Mittag-Leffler kernel \eqref{eq:kernel_ML}. Introduce a small parameter $\epsilon$ and scale the kernel as
\[\varphi_\epsilon(x) = \tfrac{1}{\epsilon} x^{\alpha-1} E_{\alpha,\alpha}(-(\lambda  + \tfrac{\wt{\gamma}}{\epsilon}) x^\alpha).\]
In terms of its Laplace transform, this scaling becomes $\wh{\varphi}_\epsilon(z) = 1 / (\epsilon(z^\alpha  + \lambda) + \wt{\gamma})$. In particular, we have
\[\wt{\gamma} \int_0^\infty \varphi_\epsilon(x) dx = \wt{\gamma} \wh{\varphi}(0) = \frac{\wt{\gamma}}{\epsilon \lambda + \wt{\gamma}} \to 1,\]
i.e. as $\epsilon \to 0$ the stability condition of the Hawkes process approaches the critical value $1$ (hence `nearly unstable'). 
From \eqref{eq:resolvent_laplace}, the Laplace transform of the resolvent kernel $\kappa_\epsilon(x)$  can be determined as
\[\wh{\kappa}_\epsilon(z) = \frac{\wh{\varphi}_\epsilon(z)}{1 - \wh{\varphi}_\epsilon(z)} = \frac{1/\epsilon}{z^\alpha + \lambda}\]
and thus the resolvent kernel is given by
\[\kappa_\epsilon(x) = \tfrac{1}{\epsilon} x^{\alpha-1} E_{\alpha,\alpha}(- \lambda x^\alpha) =  \tfrac{1}{\epsilon}\kappa(x).\]
Together with square-root scaling of the jump size we are exactly in the setting of \eqref{eq:AFI_scale} and conclude from Proposition~\ref{prop:convergence} that the (univariate) marginal distributions of $X$ converge to those of the corresponding AFV model, i.e. the rough Heston model (cf. Example~\ref{ex:rough}). Theorem~\ref{thm:multi_conv} below strengthens this result to convergence of all finite-dimensional marginal distributions.
\hfill $\diamond$
\end{example}

\subsection{The joint moment generating function}
In this subsection, we derive results on the joint moment generating function of log-price and forward variance and of the finite-dimensional marginal distributions of $X$. 
\begin{ass}\label{ass:jmgf}
We assume that $(X,\xi)$ is either an AFV model or an AFI model, and we write $R(u,w)$ for the corresponding function $R_V(u,w)$ or $R_\lambda(u,w)$. In addition we denote, for any $u \in (0,1)$, by $w_*(u)$ the unique root where
\[R(u,w_*(u)) = 0, \quad \text{and} \quad w_*(u) < 0.\]
\end{ass}
Note that the function $R(u,w)$ has already been studied in the context of affine stochastic volatility models in \cite[Lem.~3.2ff]{keller-ressel2011moment}. In particular, we note that $R(u,w)$ and $w_*(u)$ are convex functions for $u \in [0,1], w \le 0$ and continuously differentiable on the interior of their domain.

\begin{prop}\label{prop:jmgf}
Let $(X,\xi)$ be an AFV or an AFI model and let $R(u,w)$ and $w_*(u)$ be defined as in Assumption~\ref{ass:jmgf}. Let $\Delta > 0$, $T' = T + \Delta$, and let $h$ be a piecewise continuous $\Rminus$-valued function on $[0,\Delta]$, such that $w_*(u) < \int_0^\Delta \kappa(\Delta-s)h(s)ds$. Then 
\begin{multline}\label{eq:joint_affine}
\Econd{\exp\left(u X_T + \int_T^{T'} h(T'-s)\xi_T(s)ds\right)}{\cF_t} = \\
= \exp \left(uX_t + \int_t^{T'} g(T'-s,u,h) \xi_t(s)ds\right),
\end{multline}
where $g(.,u,h): \Rplus \to \Rminus$ is the unique continuous solution of the (generalized) convolution Riccati equation
\begin{align}\label{eq:Riccati_joint}
g(t,u,h) &= R\Big(u, \int_0^t \kappa(t-s) g(s,u,h) ds\Big), \quad t \in [\Delta,\infty)\\
\intertext{with initial condition}
g(t,u,h) &= h(t), \qquad t \in [0,\Delta).\label{eq:ic}
\end{align}
\end{prop}
\begin{rem}
Note that the expression \eqref{eq:joint_affine} for the joint moment generating function does not correspond to the exponential-affine transform formula (4.6) of \cite{jaber2017affine}.  Specifically, $h$ constant in  \eqref{eq:joint_affine} would give the joint moment generation function of $X_T$ and the forward variance swap $\int_T^{T'}\,\xi_T(s)\,ds$.  In contrast, $f$ constant in (4.6) of \cite{jaber2017affine} would give the the joint moment generation function of $X_T$ and quadratic variation $\int_0^{T}\,V_s\,ds$.
\end{rem}
\begin{proof}The existence of a unique, $\Rminus$-valued solution to \eqref{eq:Riccati_joint} with initial condition \eqref{eq:ic} follows from an application of Corollary~\ref{cor:Volterra_ic} with $H_u(w) = R(u,w)$. In the proofs of Theorem~\ref{thm:main} and Theorem~\ref{thm:AFI}, we have already established that $H_u$ satisfies the necessary conditions to apply the corollary. Next, we define $G_t^\Delta = \int_t^{T'}g(T'-s,u,h)\xi_t(s)ds$ and specialize to the forward variance case. By Lemma~\ref{lem:G}, it holds that
\[dG_t^\Delta = -g(T'-t,u,h)V_t dt + \left( \int_t^{T'}  \kappa(r-t) g(T'-r,u,h)dr \right)\sqrt{V_t} d W_t\]
and that $M^\Delta_t$ is a local martingale on $[0,T]$ if
\[g(T'-t,u,h) = R\left(u, \int_t^{T'} \kappa(r-t) g(T'-r,u,h) dr\right).\]
Setting $\tau = T'-t \in [\Delta, T']$ this is exactly \eqref{eq:Riccati_joint}. We conclude that $M^\Delta_t$ is a local martingale on $[0,T]$, and -- repeating the argument following \eqref{eq:Lp_bound} -- even a true martingale. Using the initial condition \eqref{eq:ic}, we observe that
\begin{multline*}
\Econd{\exp\left(u X_T + \int_T^{T'} h(T'-s)\xi_T(s)ds\right)}{\cF_t} = \Econd{M_T^\Delta}{\cF_t} = \\
= M_t = \exp \left(uX_t + \int_t^{T'} g(T'-s,u,h) \xi_t(s)ds\right),
\end{multline*}
showing \eqref{eq:affine}.
The proof in the AFI case is analogous with the following modifications: $W_t$ has to be substituted by the pure-jump martingale $\tilde{J}_t^X$ and $V_t$ by the intensity $\lambda_{t-}$. It\^o's formula for jump processes can then be applied as in the proof of Theorem~\ref{thm:AFI}.
\end{proof}

\begin{prop}\label{prop:multi_mgf}
Let $(X,\xi)$ be an AFV or an AFI model and let $R(u,w)$ and $w_*(u)$ be defined as in Assumption~\ref{ass:jmgf}.  Let $t_0 \le t_1 \le \dotsm t_n = T$ and $u = (u_0, \dotsc, u_{n-1}) \in (0,1)^n$ be such that $w_*(u_0) \le w_*(u_2) \le \dotsm w_*(u_{n-1})$. 
Then, for all $k \in \set{0, \dotsc, n-1}$, 
\begin{multline}\label{eq:multi_affine}
\Econd{\exp\left(u_{k} (X_{t_{k+1}} - X_{t_{k}}) + \dotsm + u_{n-1} (X_{T} - X_{t_{n-1}})\right)}{\cF_{t_k}} = \\
= \exp \left(\int_{t_k}^{T} g_{k}(T-s,u) \xi_{t_k}(s)ds\right),
\end{multline}
where the functions $g_k$ are defined by backward recursion as the solutions of the convolution Riccati equations
\begin{align}\label{eq:Riccati_multi}
g_{k}(t,u) &= R\Big(u_k, \int_0^t \kappa(t-s) g_k(s,u) ds\Big), \quad t \in [T - t_{k+1}, T- t_k)\\
\intertext{with initial conditions}
g_k(t,u) &= g_{k+1}(t,u), \qquad t \in [0,T- t_{k+1}).\label{eq:ic_multi}
\end{align}
\end{prop}
\begin{rem}
Note that for $k = n-1$ equation  \eqref{eq:Riccati_multi} becomes a (generalized) convolution Riccati equation \emph{without} initial condition and \eqref{eq:ic_multi} becomes void (i.e. a condition over an empty set).
\end{rem}
\begin{proof}We show the result by backward induction on $k$: For $k = n-1$ the proposition is equivalent to Theorem~\ref{thm:main}, when $(X,\xi)$ is an  AFV model, and to Theorem~\ref{thm:AFI}, when $(X,\xi)$ is an AFI model. Setting $\Delta_k := T - t_k$, we obtain from \eqref{eq:Delta_estim} in Corollary~\ref{cor:Volterra_ic} that 
\begin{equation}\label{eq:wstar_bound}
w_*(u_{n-1}) < \int_0^{\Delta_{n-1}} \kappa(\Delta_{n-1}- s)g_{n-1}(s,u)ds.
\end{equation}
For the induction step assume that \eqref{eq:multi_affine} has been shown for a certain $k$ and that \eqref{eq:wstar_bound} holds with $n-1$ replaced by $k$. Writing
\[Z_{k-1} := \exp\left(u_{k-1} (X_{t_k} - X_{t_{k-1}}) + \dotsm + u_n (X_{T} - X_{t_{n-1}})\right)\]
and applying the tower law of conditional expectations, we have
\begin{multline*}
\Econd{Z_{k-1}}{\cF_{t_{k-1}}} = \Econd{\exp\left(u_{k-1} (X_{t_k} - X_{t_{k-1}})\right) \cdot \Econd{Z_{k}}{\cF_{t_{k}}}}{\cF_{t_{k-1}}} =  \\
 = \Econd{\exp\left(u_{k-1} (X_{t_k} - X_{t_{k-1}}) + \int_{t_{k}}^{T} g_{k}(T-s,u) \xi_{t_{k}}(s)ds\right)}{\cF_{t_{k-1}}}.
\end{multline*}
Since
\[w_*(u_{k-1}) \le w_*(u_k) < \int_0^{\Delta_k} \kappa(\Delta_k- s)g_{k}(s,u)ds\]
we may apply Proposition~\ref{prop:jmgf} with $\Delta_k$ and obtain \eqref{eq:multi_affine} with $g_{k-1}$ as solution of \eqref{eq:Riccati_multi} with initial condition~\eqref{eq:ic_multi}. Finally, \eqref{eq:wstar_bound} holds with $n-1$ replaced by $k-1$, using the estimate \eqref{eq:Delta_estim} from Corollary~\ref{cor:Volterra_ic}.
\end{proof}

\subsection{Convergence of finite-dimensional marginal distributions}

\begin{thm}\label{thm:multi_conv}
Let $(X^\epsilon, \xi^\epsilon)$ be the rescaled AFI model \eqref{eq:AFI_scale}, define $c, \rho$ as in Lemma~\ref{lem:scale} and set $\kappa_V(x) = c\kappa(x)$. Then, for any $n \in \NN$ and $0 = t_0 \le t_1 \le \dotsm t_n  =T$, 
\begin{equation}\label{eq:multi_dist_conv}
\left(X_{t_0}^\epsilon, \dotsc, X_{t_n}^\epsilon\right) \xrightarrow{\epsilon \to 0} \left(X_{t_0}, \dotsc, X_{t_n} \right) \qquad \text{in distribution},
\end{equation}
where $(X, \xi)$ is the AFV model with correlation parameter $\rho$ and kernel $\kappa_V$.
\end{thm}
\begin{proof}By Lemma~\ref{lem:scale} $\tfrac{1}{\epsilon}R^\epsilon(u,w)$ converges to $R_V(u,cw)$ and the same holds true for the partial derivatives with respect to $u$ and $w$. Therefore, by the implicit function theorem, also $w^\epsilon_*(u)$ and $\pd{}{u}w^\epsilon_*(u)$ converge to $\tfrac{1}{c}w_*(u)$ and $\tfrac{1}{c}w'_*(u)$ as $\epsilon \to 0$ for all $u \in (0,1)$. Moreover, since the $w_*^\epsilon$ are convex functions of $u$, the convergence is uniform on compacts (cf. \cite[Thm.~10.8]{rockafellar1970convex}). The limit $\tfrac{1}{c}w_*(u)$ can be calculated explicitly and is given by 
\[\tfrac{1}{c}w_*(u) = \tfrac{1}{c} \left(-\rho u + \sqrt{\rho^2u^2 + (u - u^2)}\right).\]
It is easy to see that $w_*$ is decreasing on $(0,u_*)$ and increasing on $(u_*,1)$, where 
\[u_* := \begin{cases}\frac{1}{2} \frac{1 - |\rho|}{1 - \rho^2} \qquad &\text{if }\rho \in (0,1)\\ \frac{1}{4} \qquad &\text{if } |\rho| = 1. \end{cases}\]
We conclude that there is $N \in \NN$ and a closed interval $I \subset (0,u_*)$ with non-empty interior, such that $u \mapsto w_*(u)$ and $u \mapsto w_*^\epsilon(u)$ are decreasing on $I$ for all $\epsilon \le 1/N$. Introduce the set
\[D := \set{u \in I^n: u_0 \ge u_2 \ge \dotsm \ge u_{n-1}} \subset (0,1)^n\]
and note that also $D$ is closed with non-empty interior. In addition, $w^\epsilon_*(u_0) \le w^\epsilon_*(u_2) \le \dotsm w^\epsilon_*(u_{n-1})$ for all $u = (u_0, \dotsc, u_{n-1}) \in D$ and $\epsilon \le 1/N$, and the same holds for $w_*$. From 
Proposition~\ref{prop:multi_mgf} we conclude that the joint moment generating function of the increments $(X^\epsilon_{t_1} - X^\epsilon_{t_0}, X^\epsilon_{t_2} - X^\epsilon_{t_1}, \dotsc, X^\epsilon_{t_n} - X^\epsilon_{t_{n-1}})$ is of the form 
\begin{multline*}
Z^\epsilon(u) := \E{\exp\left(u_{0} (X^\epsilon_{t_{1}} - X^\epsilon_{t_{0}}) + \dotsm + u_{n-1} (X^\epsilon_{T} - X^\epsilon_{t_{n-1}})\right)} = \\
= \exp \left(\int_{0}^{T} g^\epsilon_{0}(T-s,u) \xi_{t_0}(s)ds\right),
\end{multline*}
for any $u \in D$ and $\epsilon \le 1/N$, where $g^\epsilon_0$ satisfies the iterated Riccati convolution equations \eqref{eq:Riccati_multi} with $R(u,w) = R^\epsilon(u,w)$. By Corollary~\ref{cor:Volterra_ic} each of these equations can be transformed into a non-linear Volterra equation, whose solution depends continuously on $(\epsilon,t,u)$ by \cite[Thm.~13.1.1]{gripenberg1990volterra}. In addition, Lemma~\ref{lem:scale} yields the convergence $\tfrac{1}{\epsilon}R^\epsilon(u,w) \to R_V(u,cw)$. Hence we conclude  --  as in the proof of Proposition~\ref{prop:convergence} -- that $\tfrac{1}{\epsilon}g_0^\epsilon(t,u)$ converges, uniformly for $(t,u)$ in compacts, to $g_0(t,u)$ as $\epsilon \to 0$, where $g_0(t,u)$ is the unique solution of the iterated Riccati convolution equations \eqref{eq:Riccati_multi} with $R(u,w) = R_V(u,cw)$. Consider now the joint moment generating function $Z(u)$ of the increments $(X_{t_1} - X_{t_0}, X_{t_2} - X_{t_1}, \dotsc, X_{t_n} - X_{t_{n-1}})$ of the AFV model with parameter $\rho$ and kernel $\kappa_V = c \kappa$. The convergence $\tfrac{1}{\epsilon}g_0^\epsilon(t,u) \to g_0(t,u)$ together with Proposition~\ref{prop:multi_mgf} yields
\[Z^\epsilon(u) = \exp \left(\int_{0}^{T} g^\epsilon_{0}(T-s) \xi^\epsilon_{0}(s)ds\right) \to \exp \left(\int_{0}^{T} g_{0}(T-s,u) \xi_{0}(s)ds\right) = Z(u)\]
for all $u \in D$. By \cite[Thm.~29.4 and Prob.~30.4]{billingsley1986probability} convergence of the moment generating function on a set with non-empty interior implies convergence in distribution, and \eqref{eq:multi_dist_conv} follows. 
\end{proof}

\begin{appendix}
\normalsize
\section{Some results on Volterra equations with convex non-linearity}\label{app:volterra}

We show some results on Volterra equations with convex non-linearity, of the type appearing in Theorem~\ref{thm:main} and \ref{thm:AFI}. On the non-linearity we impose the following assumptions:
\begin{ass}\label{ass:convex}
The function $H: (-\infty, \wmax] \to \RR$ is continuously differentiable and convex with a unique root $H(w_*) = 0$ in $(-\infty,\wmax]$. Moreover, $H'(w_*) < 0$ and $H(\wmax) < 0$.
\end{ass}
For a function $H$ satisfying Assumption~\ref{ass:convex}, we set
\[w_0 = \argmin_{w \in (-\infty,\wmax]} H(w);\]
if the minimum is not unique (i.e., if $H$ has a flat part), then $w_0$ shall denote the leftmost minimizer. Note that either
\begin{itemize}
\item $w_0 = \wmax$, in which case $H$ is strictly decreasing on $(-\infty, \wmax]$; or
\item $w_0 < \wmax$, in which case $H$ is strictly decreasing on $(-\infty, w_0)$ and increasing on $[w_0,\wmax]$.
\end{itemize}
In any case, $w_* < w_0 \le \wmax$ holds true. Also the following definition will be useful:
\begin{defn}\label{def:dec_env}
Let $H$ be a function satisfying Assumption~\ref{ass:convex}. The \emph{decreasing envelope} of $H$ is defined as
\begin{equation}\label{eq:dec_env}
\overline{H} := \begin{cases}H(w), \qquad &w \le w_0\\H(w_0) ,\qquad &w \in [w_0,\wmax] \end{cases}.
\end{equation}
\end{defn}
Clearly $\overline{H}$ also satisfies Assumption~\ref{ass:convex}, but is in addition decreasing and satisfies $\overline{H} \le H$.  Both Assumption~\ref{ass:convex} and Definition~\ref{def:dec_env} are illustrated in Figure~\ref{fig:H}.

\begin{figure}[htbp]
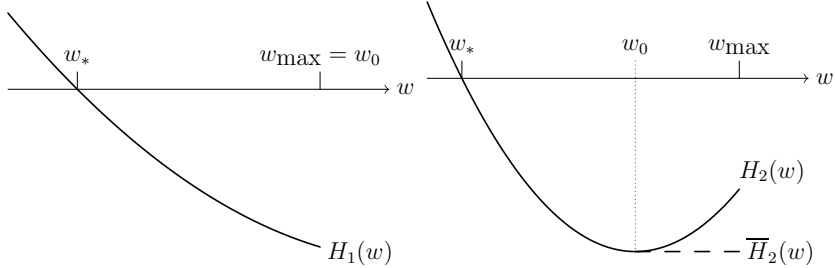

\begin{center}
\includegraphics[width=0.45\textwidth]{H_illu2.pdf}
\includegraphics[width=0.45\textwidth]{H_illu.pdf}
\caption{Illustration of two convex functions $H_1, H_2$ satisfying Assumption~\ref{ass:convex}. While $H_1$ is monotone decreasing, $H_2$ is not, and its decreasing envelope $\overline{H}_2$ is also shown.}
\label{fig:H}
\end{center}
\end{figure}

\begin{lem}\label{lem:H_prelim} Let $H: (-\infty, \wmax] \to \RR$ be a convex function that satisfies Assumption~\ref{ass:convex}; in particular it has a root $H(w_*) = 0$. Then
\begin{enumerate}[(a)]
\item For any $a \in (w_*,\wmax]$ the function
\begin{equation}\label{eq:Q1_def}
w \mapsto Q_1(w,a) = -\int_w^a\frac{d\zeta}{H(\zeta)},
\end{equation}
maps $(w_*,a]$ onto $[0,\infty)$; is strictly decreasing, and has an inverse $Q_1^{-1}(r,a)$, which maps $[0,\infty)$ onto $(w_*,a]$. 
\item For any $a \in (-\infty, w_*)$ the function
\begin{equation}\label{eq:Q2_def}
w \mapsto Q_2(w,a) =  \int_a^{w} \frac{d\zeta}{H(\zeta)},
\end{equation}
maps $[a,w_*)$ onto $[0,\infty)$; is strictly increasing, and has an inverse $Q_2^{-1}(r,a)$, which maps $[0,\infty)$ onto $[a,w_*)$. 
\end{enumerate}
\end{lem}
\begin{rem}Analogous to \eqref{eq:Q1_def}, we denote by $\overline{Q}_1$ the function
\begin{equation}\label{eq:Q1check_def}
w \mapsto \overline{Q}_1(w,a) = -\int_w^a\frac{d\zeta}{\overline{H}(\zeta)},
\end{equation}
where $\overline{H}$ is the decreasing envelope of $H$.
\end{rem}
\begin{proof}To show (a), note that the integrand $-1/H(\zeta)$ is strictly positive on $(w_*,a)$. It follows that $Q_1(.,a)$ is strictly decreasing and maps $(w_*,a]$ \emph{into} $[0,\infty)$. It remains to show that range of this map covers all of $[0,\infty)$. To this end, observe that by convexity we have
\begin{equation}\label{eq:H_convex}
H(w) \ge H'(w_*) (w - w_*), \qquad \text{for all} \; w \in (-\infty,\wmax],
\end{equation}
and $H'(w_*) < 0$. Thus, we obtain
\[\lim_{w \downarrow w_*}Q_1(w,a) = -\int_{w_*}^a \frac{d\zeta}{H(\zeta)} \ge -\frac{1}{H'(w_*)} \int_{w_*}^a \frac{d\zeta}{\zeta - w_*} = +\infty.\]
The proof of (b) is analogous; only the different sign of $H$ on $(-\infty,w_*)$ has to be taken into account.
\end{proof}

\begin{thm}\label{thm:Volterra}
Let $\kappa$ be an $L_1$-kernel in the sense of Definition~\ref{def:kernel} and let $H$ be a convex function that satisfies Assumption~\ref{ass:convex}; in particular $H(w_*) = 0$ is its unique root in $(-\infty,\wmax]$. For any continuous function $a: \Rplus \to (-\infty,\wmax]$ consider the non-linear Volterra equation
\begin{equation}\label{eq:Volterra}
f(t) = a(t) + \int_0^t \kappa(t-s) H(f(s)) ds, \qquad t \in \Rplus.
\end{equation}
\begin{enumerate}[(a)]
\item If $a$ is increasing with values in $(w_*,w_0]$ then \eqref{eq:Volterra} has a unique global solution $f$ which satisfies
\begin{equation}\label{eq:f_estim}
w_* < r_1(t) \le f(t) < a(t), \qquad \forall\,t  > 0,
\end{equation}
where $r_1(t) = Q_1^{-1}\left(\int_0^t \kappa(s) ds, a(0)\right)$ and $Q_1$ is given by \eqref{eq:Q1_def}.
\item If $a \equiv w_*$ then $f \equiv w_*$ is the unique global solution of \eqref{eq:Volterra}
\item If $a$ is decreasing with values in $(-\infty,w_*)$ then \eqref{eq:Volterra} has a unique global solution $f$ which satisfies
\begin{equation}\label{eq:f_estim2}
a(t) < f(t) \le r_2(t) < w_*, \qquad \forall\,t  > 0,
\end{equation}
where $r_2(t) = Q_2^{-1}\left(\int_0^t \kappa(s) ds,a(0)\right)$ and $Q_2$ is given by \eqref{eq:Q2_def}.
\end{enumerate}
In addition, case (a) can be extended to the following more general statement:
\begin{enumerate}[(a')]
\item If $a$ is increasing with values in $(w_*,\wmax]$ then \eqref{eq:Volterra} has a unique global solution $f$ which satisfies
\begin{equation}\label{eq:f_estim3}
w_* < \overline{r}_1(t) \le f(t) < a(t), \qquad \forall\,t  > 0,
\end{equation}
where $\overline{r}_1(t) = \overline{Q}_1^{-1}\left(\int_0^t \kappa(s) ds, a(0)\right)$ and $\overline{Q}_1$ is given by \eqref{eq:Q1check_def}.
\end{enumerate}
\end{thm}
\begin{rem}Clearly, if $H$ is decreasing (and hence $w_0 = \wmax$), cases (a) and (a') coincide. In the general case (a) gives better bounds on $f$ than (a'), but is more restrictive in its assumption on the function $a$. 
\end{rem}

Before proving the theorem, we add two Corollaries that are used in the proofs of Theorems~\ref{thm:main}, \ref{thm:AFI} and \ref{thm:multi_conv}.

\begin{cor}\label{cor:Volterra}
Under the assumptions of Theorem~\ref{thm:Volterra}, consider the non-linear integral equation
\begin{equation}\label{eq:Volterra_outer}
g(t) = H\left(a(t) + \int_0^t \kappa(t-s) g(s) ds\right), \qquad t \in \Rplus.
\end{equation}
\begin{enumerate}[(a)]
\item If $a$ is increasing with values in $(w_*,w_0]$ then \eqref{eq:Volterra_outer} has a unique global solution $g$ which satisfies
\begin{equation}\label{eq:g_estim}
H(a(t)) < g(t) \le H(r_1(t)) < 0, \qquad \forall\,t  > 0.
\end{equation}
\item If $a \equiv w_*$ then $g \equiv 0$ is the unique global solution of \eqref{eq:Volterra_outer}.
\item If $a$ is decreasing with values in $(-\infty,w_*)$ then \eqref{eq:Volterra_outer} has a unique global solution $g$ which satisfies
\begin{equation}\label{eq:g_estim2}
0 < g(t) \le H(r_2(t)) < H(a(t)), \qquad \forall\,t  > 0.
\end{equation}
\end{enumerate}
In addition, case (a) can be extended to: 
\begin{enumerate}[(a')]
\item If $a$ is increasing with values in $(w_*,\wmax]$ then \eqref{eq:Volterra_outer} has a unique global solution $g$ which satisfies
\begin{equation}\label{eq:g_estim3}
g(t) < 0, \qquad \forall\,t  > 0.
\end{equation}
\end{enumerate}
In any of the above cases, $g(t) = H(f(t))$, where $f$ is the solution of \eqref{eq:Volterra}.
\end{cor}

\begin{cor}\label{cor:Volterra_ic}
Let the assumptions of Theorem~\ref{thm:Volterra} hold with $\wmax = 0$. Let $\Delta > 0$ and let $h$ be a piecewise continuous function from $[0,\Delta)$ to $\Rminus$. Consider the non-linear integral equation
\begin{align}\label{eq:Volterra_init}
g(t) &= H\left(\int_0^t \kappa(t-s) g(s) ds\right), \qquad t \in [\Delta, \infty),
\intertext{with initial condition}
g(t) &= h(t), \qquad t \in [0,\Delta).\notag
\end{align}
If $w_* <  \int_0^\Delta \kappa(\Delta - s) h(s)ds$, then \eqref{eq:Volterra_init} has a unique global solution $g$ taking values in $\Rminus$, which satisfies  
\begin{equation}\label{eq:Delta_estim}
w_* <  \int_0^t \kappa(t - s) g(s)ds \quad \text{for all $t \ge 0$.}
\end{equation} 
\end{cor}

We start with the proof of Theorem~\ref{thm:Volterra}, which closely follow the account of Lakshmikantham's comparison method in \cite[Sec.~II.7]{bainov2013integral}.
\begin{proof}[Proof of Theorem~\ref{thm:Volterra}]
Clearly, $H$ can be extended to a continuous function on all of $\RR$ and thus it follows from \cite[Thm.~12.1.1]{gripenberg1990volterra} that \eqref{eq:Volterra} has a local continuous solution $f$ on an interval $[0,\Tmax)$ with $\Tmax > 0$. In addition, $\Tmax$ can be chosen maximal, in the sense that the solution cannot be continued beyond $[0,\Tmax)$.\\
\textbf{Case (a):} By assumption, $a$ is increasing and takes values in $(w_*,w_0]$. Set
\begin{equation}\label{eq:T0_def}
T_* := \inf \set{t \in (0,\Tmax): f(t) = w_* \text{ or } f(t) = a(\Tmax)}
\end{equation}
and note that $T_* >  0$. From \eqref{eq:Volterra} it is clear that 
\begin{equation}\label{eq:ineq_1}
f(t) = a(t) + \int_0^t \kappa(t-s) H(f(s)) ds < a(t) \le a(\Tmax), \quad \forall t \in [0,T_*),
\end{equation}
i.e. the lower bound $w_*$ in \eqref{eq:T0_def} is always hit before the upper bound $a(\Tmax)$. In addition, using that the kernel $\kappa$ is decreasing, we obtain that
\begin{equation}\label{eq:v_def}
f(t) = a(t) + \int_0^t \kappa(t-s) H(f(s)) ds \le a(T) + \int_0^t \kappa(T-s) H(f(s)) ds := v(t,T)
\end{equation}
for all $0 \le t \le T \le T_*$. The function $v(t,T)$, which we have just defined,  satisfies
\begin{align}
v(t,t) &= f(t)\label{eq:v_at_t}\\
v(0,T) &= a(T) \ge a(0)\label{eq:v_initial}
\end{align}
and the differential inequality 
\begin{equation}\label{eq:diff_ineq}
\pd{}{t} v(t,T) = \kappa(T-t) H(f(t)) \ge \kappa(T-t) H(v(t,T)).
\end{equation}
Here, we have used \eqref{eq:v_def} and the fact that $H$ is decreasing on $(w_*,w_0]$. Together with the initial estimate \eqref{eq:v_initial}, a standard comparison principle for differential inequalities (cf. \cite[II.\S 9]{walter1996ordinary}) yields
\begin{align}
v(t,T) &\ge r(t,T), \label{eq:ineq_v}\\
\intertext{where}
\pd{}{t}r(t,T) &= \kappa(T-t) H(r(t,T)), \qquad r(0,T) = a(0).\label{eq:diffeq_r}\end{align}
We claim that the differential equation \eqref{eq:diffeq_r} is solved by 
\begin{equation}\label{eq:Q1_sol}
r(t,T) = Q_1^{-1}\left(\int_0^t \kappa(T-s) ds,a(0)\right).
\end{equation}
Indeed, applying $Q_1(.,a(0))$ to both sides of \eqref{eq:Q1_sol} yields
\[\int_0^t \kappa(T-s)ds = Q_1(r(t,T),a(0)) = - \int_{r(t,T)}^{a(0)} \frac{d\zeta}{H(\zeta)}.\]
Taking $\pd{}{t}$-derivatives, we obtain
\[\kappa(T-t) = \frac{1}{H(r(t,T))} \pd{}{t} r(t,T)\]
which is equivalent to \eqref{eq:diffeq_r}. From \eqref{eq:v_def}, \eqref{eq:v_at_t} and \eqref{eq:ineq_v} we obtain the bound
\begin{equation}\label{eq:lower_bd_f}
r_1(t) := \lim_{T \downarrow t} r(t,T) \le \lim_{T \downarrow t} v(t,T) = f(t)
\end{equation}
for all $t \inÊ[0,T_*)$. This implies that
\begin{equation}\label{eq:final_comp}
\lim_{t \to T_*} f(t) \ge r_1(T_*) > w_*, 
\end{equation}
which, in light of \eqref{eq:T0_def}, means that $T_* = \Tmax$, i.e. we have shown the bounds \eqref{eq:f_estim} to hold for all $t \in [0,\Tmax)$. However, by \cite[Thm.~12.1.1]{gripenberg1990volterra} $\lim_{t \to \Tmax} |f(t)| = \infty$ whenever $\Tmax < \infty$. We conclude that $\Tmax = \infty$, and hence that $f$ is a global solution of \eqref{eq:Volterra}. Uniqueness follows from \cite[Thm.~13.1.2]{gripenberg1990volterra}.\\
\textbf{Case (b):} By assumption, $a  \equiv w_*$. Since $H(w_*) = 0$ it is clear that $f(t) \equiv w_*$ is a global solution of \eqref{eq:Volterra}. Uniqueness follows from \cite[Thm.~13.1.2]{gripenberg1990volterra}.\\
\textbf{Case (c):} By assumption, $a$ is decreasing and takes values in $(-\infty, w_*]$. This case can be handled analogous to case (a) with the following adaptations: The inequality signs in equations \eqref{eq:ineq_1} -- \eqref{eq:ineq_v} have to be reversed. In \eqref{eq:Q1_sol} $Q_1$ has to be substituted by $Q_2$ and also in \eqref{eq:lower_bd_f} and \eqref{eq:final_comp} the inequalities have to be reversed.\\
\textbf{Case (a'):} The proof of Case (a) applies, except for the following modification: \eqref{eq:diff_ineq} holds only when $v(t,T) \le w_0$, since $H$ is decreasing only on $(-\infty,w_0]$. However, when $v(t,T) > w_0$, we can use the trivial estimate
\[\pd{}{t} v(t,T) = \kappa(T-t) H(f(t)) \ge \kappa(T-t) H(w_0),\]
which can be combined with \eqref{eq:diff_ineq} into
\[\pd{}{t} v(t,T) = \kappa(T-t) H(f(t)) \ge \kappa(T-t) \overline{H}(v(t,T)),\]
where $\overline{H}$ is the decreasing envelope of $H$ from Definition~\ref{def:dec_env}. The remaining proof of Case (a) applies after substituting $H$ by $\overline{H}$ and $Q_1$ by $\overline{Q}_1$.
 \end{proof}

\begin{proof}[Proof of Corollary~\ref{cor:Volterra}] Let $f$ be the global solution of \eqref{eq:Volterra}. Applying $H$ to both sides of \eqref{eq:Volterra}, we see that $g(t) := H(f(t))$ is a global solution of \eqref{eq:Volterra_outer}. Next, we show uniqueness. To this end, assume that $\wt{g}$ is a local solution of \eqref{eq:Volterra_outer} on $[0,T)$, different from $g$ and define
\[\wt{f}(t) := a(t) + \int_0^t \kappa(t-s) \wt{g}(s)ds.\]
Clearly, $\wt{g}(t) = H(\wt{f}(t))$ on $[0,T)$, and hence $\wt{f}$ is a local solution of \eqref{eq:Volterra}. By \cite[Thm.~13.1.2]{gripenberg1990volterra}, this solution is unique, and we conclude that $\wt{f} = f$, and hence also $\wt{g} = g$. Finally, applying $H$ --  which is decreasing on $(-\infty,w_0]$ -- to the inequalities \eqref{eq:f_estim} and \eqref{eq:f_estim2} yields \eqref{eq:g_estim} and \eqref{eq:g_estim2}. In case (a') monotonicity of $H$ is lost, but $H(w) < 0$ for all $w \in (w_*,\wmax]$ yields \eqref{eq:g_estim3}.
\end{proof}

\begin{proof}[Proof of Corollary~\ref{cor:Volterra_ic}] Set
\[a(t) := \int_0^\Delta \kappa(t + \Delta - s)h(s)ds\]
and note that $a$ is increasing with values in $(w_*,0]$. Consider the non-linear Volterra equation
\begin{equation}\label{eq:f_init}
f(t) = a(t) + \int_0^t \kappa(t-s) H(f(s))ds,
\end{equation}
which has a unique global solution $f$ by Theorem~\ref{thm:Volterra}.(a) or (a'). For $t' \in \Rplus$ set
\[g(t') = \begin{cases} H(f(t'-\Delta)), &\qquad t' \in [\Delta, \infty)\\ h(t') &\qquad t' \in [0,\Delta) \end{cases}.\]
For $t' \ge \Delta$ we have
\begin{align*}
g(t') &= H(f(t'-\Delta)) = H\left(\int_0^\Delta \kappa(t' - s) h(s) ds + \int_\Delta^{t'} \kappa(t'-s)g(s)ds\right) = \\
&= H\left(\int_0^{t'} \kappa(t'-s)g(s)ds\right),
\end{align*}
showing that $g$ is a global solution of \eqref{eq:Volterra_init}. From cases (a) or (a') of Theorem~\ref{thm:Volterra}, we obtain the bound
\[w_* < f(t' - \Delta) = \int_0^{t'} \kappa(t'-s)g(s)ds,\]
as claimed. To show uniqueness, assume that $\wt{g}$ is a solution of \eqref{eq:Volterra_init}, different from $g$. Setting
\[\wt{f}(t) := a(t) + \int_0^t \kappa(t-s) \wt{g}(s+\Delta),\]
we see that $\wt{f}$ is a solution of \eqref{eq:f_init} and conclude from Theorem~\ref{thm:Volterra} that $\wt{f} = f$ and hence also $\wt{g} = g$. 
\end{proof}

\section{Theorem~\ref{thm:main} in the case $\rho > 0$}\label{app:proof}
We provide the remaining part of the proof of Theorem~\ref{thm:main} in the case $\rho > 0$. Our starting point is the quadratic equation \eqref{eq:quadratic}, which has been obtained without any assumption on the sign of $\rho$. In the case $\rho > 0$, additional arguments are needed, since this equation may have \emph{two} negative solutions. 
\begin{proof}[Theorem~\ref{thm:main}, `only if' part in the case $\rho > 0$.]
On the set $S = \set{(t,\omega): V_t(\omega) \neq 0}$ we set
\[k_t(\tau,\omega) = \frac{1}{V_t(\omega)}\eta_t(t+\tau,\omega),\]
for $\tau \ge 0$.  Note that by Assumption~\ref{ass:sigma} $\tau \mapsto k_t(\tau,\omega)$ must be a decreasing $L_2$-kernel for a.e. $(t,\omega)$. Since \eqref{eq:beta_sigma} holds trivially if $S$ is a $dt \otimes d\PP$-nullset, we may, without loss of generality, assume that $S$ is not a nullset and consider only $(t, \omega) \in S$ in the remainder of the proof. 
Inserting into \eqref{eq:v} yields
\[h_t(t + \tau,u) = \sqrt{V_t} \int_0^\tau g(\tau-s,u) k_t(s) ds = \sqrt{V_t} \cdot (g \star k)_t(\tau,u).\]
Plugging into \eqref{eq:quadratic} and eliminating $V_t$ gives
\begin{equation}\label{eq:quadratic_repeat}
\frac{1}{2}(u^2 - u) - g(\tau,u) + u \rho (g \star k)_t(\tau,u) + \frac{1}{2} (g \star k)_t(\tau,u)^2 = 0,
\end{equation}
with is a quadratic equation in the variable $(g \star k)_t(\tau,u)$ with two solutions
\begin{equation}\label{eq:q_pm}
q_\pm(\tau,u) = -\rho u \pm \sqrt{u^2(\rho^2 - 1) + u + 2g(\tau,u)},
\end{equation}
both of which may be negative. However, using continuity of $g(.,\tau)$ and evaluating \eqref{eq:quadratic_repeat} at $\tau \downarrow 0$ yields that $g(0,u) = \frac{1}{2}(u^2 - u)$. Inserting into \eqref{eq:q_pm} and using $(g \star k)_t(0,u) = 0$ this initially selects the solution $q_+$ and shows that
\[(g \star k)_t(\tau,u) = q_+(\tau,u), \quad \text{for all } \tau \in [0,T_*(u)),\] 
where $T_*(u)$ is the first collision time of $q_+$ and $q_-$, i.e. ,
\[T_*(u) := \inf \set{\tau > 0: g(\tau,u) = \frac{1}{2}(u^2 - u) - \frac{u^2\rho^2}{2}}.\]
On the interval $[0,T_*(u))$ we can proceed as in the case of $\rho \le 0$ and obtain that
\begin{equation}\label{eq:eta_rep}
\eta_t(t+\tau) = \sqrt{V_t} \, k_t(\tau) = \sqrt{V_t} \pd{}{\tau} \left( \int_0^\tau q_+(\tau-s,u) \pi(ds,u) \right), \quad \tau \in [0,T_*(u))
\end{equation}
where $\pi$ is the resolvent of the first kind of $g(\tau,u)$. Therefore, to complete the proof it suffices to show that $T_*(u)$ can be made arbitrarily large by choosing a suitable $u \in (0,1)$. To this end, note that \eqref{eq:quadratic_repeat} is a convolution Riccati equation for $g(.,u)$ with the kernel $k_t(.)$, i.e., 
\[g(\tau,u) = R_V(u,(g \star k)_t(\tau,u)).\]
Applying Lemma~\ref{lem:Riccati} we obtain that $g(.,u)$ is its unique continuous solution, which, by Corollary~\ref{cor:Volterra}, can be written as $g(\tau,u) = R_V(u,f(\tau,u))$, where $f(\tau,u)$ solves
\begin{equation}\label{eq:f_Riccati}
f(\tau,u) = \int_0^\tau k_t(\tau-s) R_V(u,f(s,u)) ds.
\end{equation}
Moreover, the collision time can be represented in terms of $f$ as
\begin{equation}\label{eq:collision_f}
T_*(u) := \inf \set{t > 0: f(t,u) = -\rho u}.
\end{equation}
By convexity, we can estimate $R_V(u,w)$ from below as $R_V(u,w) \ge w\rho + \tfrac{1}{2}(u^2 - u)$.
Hence, from \eqref{eq:f_Riccati} we obtain the estimate
\begin{equation}\label{eq:simple_est}
f(\tau,u) \ge \rho (k_t \star f)(\tau,u) + \tfrac{1}{2}(u^2 - u).
\end{equation}
Let $r_t$ be the $\rho$-resolvent of $k_t$, and note that $r_t$ is again an $L_2$-kernel (in particular non-negative) by Lemma~\ref{lem:resolvent}. By the generalized Gronwall Lemma of \cite[Lem.~9.8.2]{gripenberg1990volterra} it follows that $f(\tau,u) \ge l(\tau,u)$, where $l$ solves the linear Volterra equation
\[l(\tau,u) = \rho (k_t \star l)(\tau,u) + \tfrac{1}{2}(u^2 - u).\]
Moreover, using \cite[Thm.~2.3.5]{gripenberg1990volterra}, we can express $l(t,u)$ in terms of the $\rho$-resolvent $r_t$ and obtain $f(\tau,u) \ge l(\tau,u) = \tfrac{1}{2}(u^2 - u)\int_0^\tau r_t(s)$ for all $u \in (0,1)$. Combining with \eqref{eq:collision_f}, we finally obtain
\[\int_0^{T_*(u)} r_t(s) ds \ge -\frac{\rho u}{\tfrac{1}{2}(u^2 - u)} = \frac{2\rho}{1-u}.\]
Sending $u \uparrow 1$ the right-hand side can be made arbitrarily large and we conclude that $\lim_{u \uparrow 1} T_*(u) = +\infty$, which together with \eqref{eq:eta_rep} completes the proof. 
\end{proof}
\end{appendix}

\bibliographystyle{plain}
\bibliography{references}
\end{document}